\documentclass[onefignum,onetabnum]{article}
\usepackage[textsize=tiny]{todonotes}
\usepackage{mathtools}
\usepackage[colorlinks=true,allcolors=blue,linkcolor=red]{hyperref}%
\usepackage{lipsum}
\usepackage{amsfonts}
\usepackage{graphicx}
\usepackage{epstopdf}
\usepackage{cleveref}
\usepackage{algorithmic}
\ifpdf
  \DeclareGraphicsExtensions{.eps,.pdf,.png,.jpg}
\else
  \DeclareGraphicsExtensions{.eps}
\fi

\newcommand{\TheTitle}{\textbf{Low-rank multi-parametric covariance identification}}

\newtheorem{remark}{Remark}
\newtheorem{proposition}{Proposition}
\newtheorem{definition}{Definition}

\newtheorem{problem}{Problem}
\crefname{problem}{Problem}{Problems}

\crefname{result}{Result}{Results}
\newtheorem{hypothesis}{Hypothesis}
\crefname{hypothesis}{Hypothesis}{Hypotheses}

\newenvironment{proof}{\paragraph{Proof:}}{\hfill$\square$}

\title{{\TheTitle}\thanks{{This work was supported by (i) ``la Caixa'' Banking Foundation (ID 100010434) under project LCF/BQ/AN13/10280009, (ii) the Fonds de la Recherche Scientifique (FNRS) and the Fonds Wetenschappelijk Onderzoek (FWO) -- Vlaanderen under EOS
(Project 30468160), (iii) the ``Communaut\'e fran\c{c}aise de Belgique -- Actions de Recherche Concert\'ees'' (Contract ARC 14/19-060),  (iv) the WBI-World Excellence Fellowship, (v) the United States Department of Energy, Office of Advanced Scientific Computing Research (ASCR).
The second author is supported by the National Physical Laboratory and the Alan Turing Institute. Most of the work was done when the second author was with ICTEAM, UCLouvain, Belgium.}}}

\author{
Antoni Musolas\thanks{Center for Computational Science \& Engineering, Massachusetts Institute of Technology, Cambridge, MA USA (musolas@mit.edu,
    ymarz@mit.edu).}
  \and Estelle Massart\thanks{Mathematical Institute, University of Oxford, OX2 6GG, UK (estelle.massart@maths.ox.ac.uk).} \and Julien M. Hendrickx\thanks{ICTEAM Institute, UCLouvain, Avenue Georges Lema\^itre 4 bte L4.05.01, 1348 Louvain-la-Neuve, Belgium, (julien.hendrickx@uclouvain.be, pa.absil@uclouvain.be)}.  \and P.-A.\ Absil\footnotemark[4] \and
  Youssef Marzouk\footnotemark[2]
}
\usepackage{amsopn}

\usepackage[a4paper,centering]{geometry}

\def \T {\top \!}
\def \R {\mathbb{R}}
\def \argmin {\arg\!\min}

\renewcommand{\O}[1]{\mathcal{O}(#1)}
\newcommand{\psd}[2]{\mathcal{S}_+(#1,#2)}

\DeclareMathOperator{\Tr}{tr}

\usepackage{tikz}
\usepackage{multicol}
\usepackage{graphicx}

\definecolor{darkgreen}{rgb}{0,0.6,0.1}
\definecolor{darkblue}{rgb}{0,0,0.7}
\definecolor{mygreen}{rgb}{0.1,0.7,0.3}
\definecolor{myblue}{rgb}{0,0.5,1}
\definecolor{myred}{rgb}{1,0,0}
\definecolor{mygray}{rgb}{0.6,0.6,0.6}

\begin{document}
\maketitle
\vspace{-0.7cm}
\begin{abstract}
We propose a differential geometric construction for families of low-rank covariance matrices, via interpolation on low-rank matrix manifolds. In contrast with standard parametric covariance classes, these families offer significant flexibility for problem-specific tailoring via the choice of ``anchor'' matrices for the interpolation. Moreover, their low-rank facilitates computational tractability in high dimensions and with limited data. We employ these covariance families for both interpolation and identification, where the latter problem comprises selecting the most representative member of the covariance family given a data set. In this setting, standard procedures such as maximum likelihood estimation are nontrivial because the covariance family is rank-deficient; we resolve this issue by casting the identification problem as distance minimization. We demonstrate the power of these differential geometric families for interpolation and identification in a practical application: wind field covariance approximation for unmanned aerial vehicle navigation.
\end{abstract}


  \textbf{AMS:} 53C22, 62J10

\section{Introduction}\label{sec:intro}

One of the most fundamental problems in multivariate analysis and uncertainty quantification is the construction of covariance matrices. Covariance matrices are an essential tool in climatology, econometrics, model order reduction, biostatistics, signal processing, and geostatistics, among other applications.
As a specific example (which we shall revisit in this paper), covariance matrices of wind velocity fields \cite{lawrance2010simultaneous,lawrance2011path,lolla2015path,cococcioni2015adaptive} capture the relationships among wind velocity components at different points in space. These relationships enable recursive approximation or updating of the wind field as new pointwise measurements become available. Similarly, in oil reservoirs \cite{oliver2011recent,oliver1997markov}, covariance matrices allow information from borehole measurements to be propagated into more accurate global estimates of the permeability field. In telecommunications \cite{liu2019target}, covariance matrices and their eigenvectors are paramount for discerning between signal and noise.

Widely used methods in spatial statistics  \cite{cressie1999classes,ripley2005spatial,stein2004approximating} include variogram estimation \cite{cressie1980robust,cressie1990origins} (often a first step in kriging \cite{stein2012interpolation,driscoll1998consistent,guhaniyogi2019multivariate}) or tapering of sample covariance matrices \cite{guerci1999theory,furrer2006covariance}. Other regularized covariance estimation methodologies include
sparse covariance and precision matrix estimation \cite{cai2011adaptive,friedman2008sparse,cai2011constrained} and many varieties of shrinkage \cite{ledoit2004well,ledoit2012nonlinear,schafer2005shrinkage,ledoit2018optimal}. Many of these methods make rather generic assumptions on the covariance matrix (e.g., sparsity of the precision, or some structure in the spectrum); others (e.g., Mat\'{e}rn covariance kernels or particular variogram models) assume that the covariance is described within some parametric family. The latter can make the estimation problem quite tractable, even with relatively limited data, since the number of degrees of freedom in the family may be small.


Unfortunately, standard parametric covariance families (e.g., Mat\'{e}rn) \cite{cressie1992statistics,rasmussen2006gaussian,rue2005gaussian} can be insufficiently flexible for practical applications. For instance, in wind field modeling, these covariance families are  insufficiently expressive to capture the non-stationary and multiscale features of the velocity or vorticity \cite{langelaan2011wind,langelaan2012wind,larrabee2014wind}. Nonparametric approaches such as sparse precision estimation may be less restrictive, but neither approach easily allows prior knowledge---such as known wind covariances at other conditions---to be incorporated. Moreover, most of these methods yield full-rank covariance matrices, which are impractical for high-dimensional problems. For example, direct manipulation of full-rank covariances in high dimensional settings might preclude recursive estimation (i.e., conditioning) from being performed online.

In this paper, we propose to build parametric covariance functions from piecewise-B\'ezier interpolation techniques on manifolds, using representative covariance matrices as anchors. (See \cite{musolas2020geodesically} for a related approach using full-rank geodesics.) These covariance functions offer significant flexibility for problem-specific tailoring and can be reduced to any given rank; both of these features enable broad applicability. We use our proposed covariance functions for
identification---finding the most representative member of a family given a data set---and interpolation---given anchors at known or ``labeled'' conditions, predicting covariance matrices at new conditions between those of the anchors. Observe that these objectives do not require an appeal to asymptotic statistical properties, which are not investigated here.

Interpolation on matrix manifolds has been an active research topic in the last few years; see, e.g.,~\cite{Boumal2011,Samir2012,Hinkle2014,Kim2018,conti2008full}. Here we propose to rely on piecewise-B\'ezier curves and surfaces on manifolds that have been investigated, e.g., in~\cite{absil2016differentiable,Bergmann2018a,gousenbourger2018data,Modin2018}. As for the space to interpolate on, an immediate choice would have been the set of all $n\times n$ covariance matrices, namely the set of $n\times n$ symmetric positive-semidefinite (PSD) matrices; however this set is not a manifold. (Specifically, it is not a submanifold of $\mathbb{R}^{n\times n}$.) Instead, we will consider the set of $n\times n$ covariance matrices of fixed rank $r$, which is known to be a manifold; see, e.g.,~\cite{vandereycken2009embedded}. Among others,~\cite{pennec2006riemannian} and~\cite{moakher2006symmetric} investigated the full-rank case ($r=n$) and pioneered building matrix functions using geodesics on the manifold of symmetric positive-definite (SPD) matrices. For the low-rank case ($r<n$), several geometries are available~\cite{bonnabel2009riemannian,vandereycken2009embedded,vandereycken2012riemannian}. We will resort to the geometry proposed in~\cite{journee2010low} and further developed in~\cite{massart2018quotient}, as it has the crucial advantage of providing a closed-form expression for the endpoint geodesic problem, which appears as a pervasive subproblem in B\'ezier-type interpolation techniques on manifolds.

We point out that univariate interpolation on the manifold of fixed-rank PSD matrices has already been applied, e.g., to the wind field approximation problem in~\cite{Gousenbourger2017}, to protein conformation in~\cite{Li2014}, to computer vision and facial recognition in~\cite{Kacem2018, szczapa2019}, and to parametric model order reduction in~\cite{Massart2019}. The present work is, to our knowledge, a first step towards the multivariate case.

There are three original contributions in this paper. First, we devise new \emph{low-rank} parameterized covariance families, based on given problem-specific anchor covariance matrices. The rank of the anchor matrices is assumed to be equal to some value $r$, usually much smaller than the size $n$ of the matrices. The resulting covariance families are shown to contain only matrices of rank less than or equal to $r$. In high-dimensional applications, this allows reducing the computational cost of matrix manipulations, while still explaining the majority of the variance of the data. Working with low-rank covariance families also results in robustness to \emph{small} data. Second, we propose  minimization of an appropriate loss function as an alternative to maximum likelihood estimation for selecting the most representative member of this covariance family given some data. When the covariance family is low-rank,  maximum likelihood estimation is not trivial, as the probability density of the data (assumed Gaussian) is degenerate. Third, we demonstrate the previous two points in an application: wind field velocity characterization for UAV navigation. We notice that when connecting anchors labeled with different values of the prevailing wind conditions, the values in between correspond to intermediate wind conditions.

%


The rest of this paper is organized as follows.
We summarize the tools needed to work on the manifold of fixed-rank PSD matrices in \Cref{sec:ssd}. We introduce our new covariance functions in~\Cref{sec:surfaces}, for both the one- and the multi-parameter cases. In \Cref{sec:results}, we present methods to solve the covariance identification problem via distance minimization. In \Cref{sec:app}, we illustrate the behaviors of the different covariance functions on a case study: wind field approximation. 

\section{The geometry of the set of positive-semidefinite matrices}\label{sec:ssd}

In this section, we define useful tools to work on the manifold $\psd{r}{n}$ of positive-semidefinite (PSD) matrices, with rank $r$ and size $n\times n$, with $r<n$.

Several metrics have been proposed for this manifold but, to our knowledge, none of them manages to turn it into a complete metric space with a closed-form expression for endpoint geodesics. We use the quotient geometry $\psd{r}{n} \simeq \mathbb{R}_*^{n\times r} / \O{r}$ proposed in~\cite{journee2010low} and further developed in~\cite{massart2018quotient}, with $\mathbb{R}_*^{n\times r}$ endowed with the Euclidean metric. This geometry relies on the fact that a matrix $A \in \psd{r}{n}$ can be factorized as $A=Y_{A}Y_{A}^{\T}$, where the factor $Y_{A}\in \mathbb{R}_*^{n\times r}$ has full column rank. The decomposition is not unique, as each factor $\mathcal{S}_{A} = Y_A Q$, with $Q \in \O{r}$ an orthogonal matrix, leads to the same product. As a consequence, any PSD matrix $A$ is represented by an equivalence class:
\begin{displaymath}
[Y_{A}]=\{Y_{A}Q | Q \in \O{r} \}.
\end{displaymath}
In our computations, we work with representatives of the equivalence classes. For example, the geodesic between two PSD matrices $A$ and $B$ will be computed based on two arbitrary representatives $Y_A$, $Y_B$, of the corresponding equivalence classes. The geodesic will of course be invariant to the choice of the representatives. Moreover, this approach saves computational cost as the representatives are of size $n\times r$, instead of $n\times n$.

In~\cite{massart2018quotient}, the authors propose an expression for the main tools to perform computations on $\psd{r}{n}$, endowed with this geometry. The geodesic ${\varphi}_{A_1 \rightarrow A_2}$ between two PSD matrices $A_1$, $A_2$, with representatives $Y_{A_1}$ and $Y_{A_2}$, is given by:
\begin{equation}
{\varphi}_{A_1 \rightarrow A_2}(t)=Y_{{\varphi}_{A_1 \rightarrow A_2}(t)}Y_{{\varphi}_{A_1 \rightarrow A_2}(t)}^{\T} \quad \mathrm{with} \quad Y_{{\varphi}_{A_1 \rightarrow A_2}(t)}=Y_{A_1}+t\dot{Y}_{A_1 \rightarrow A_2}.
\label{eq:geodesic}
\end{equation}
In this expression, the vector $\dot{Y}_{A_1 \rightarrow A_2}$ is defined as $\dot{Y}_{A_1 \rightarrow A_2}=Y_{A_2}Q^\T-Y_{A_1}$ with $Y_{A_1}^{\T}Y_{A_2}=HQ$ a polar decomposition. In the generic case where $Y_{A_1}^\T Y_{A_2}$ is nonsingular, the polar decomposition is unique and the resulting curve $t \in [0,1] \mapsto \varphi_{A_1 \to A_2}(t)$ is the unique minimizing geodesic between $A_1$ and $A_2$. This curve has the following properties:

\begin{enumerate}
\item ${\varphi}_{A_1 \rightarrow A_2}(0)=A_1$, and ${\varphi}_{A_1 \rightarrow A_2}(1)=A_2$.
\item For each $t\in \mathbb{R}$, ${\varphi}_{A_1 \rightarrow A_2}(t)\in \psd{\leq r}{n}$,
\end{enumerate}
where the notation $\psd{\leq r}{n}$ stands for the set of positive-semidefinite matrices of rank upper-bounded by $r$.

Notice that the last property suggests that this manifold is not a complete metric space for this metric, because the points in the geodesics are not necessarily of rank $r$. However, completeness is not central for statistical modeling, and, as we will see in \Cref{sec:app}, the fact that not all covariance matrices have the same rank will not have any consequence in practice.

We finally mention that~\cite{massart2018quotient} also contains expressions for the exponential and logarithm maps, on which rely the patchwise B\'ezier surfaces introduced in~\Cref{sec:bgeodesic} below. Referring to the PSD matrices discussed above (e.g., $A_1, A_2$) as ``data matrices,'' we make  the following assumption:
\begin{hypothesis}\label{hypo:noCutLocus}
The data matrices are such that all the logarithm and exponential maps to which we refer in the sequel are well-defined.
\end{hypothesis}
\Cref{hypo:noCutLocus} will typically be satisfied when the data matrices are sufficiently close to each other; see~\cite{massart2018quotient} for more information.

Instead of working directly on the quotient manifold (which involves the computation of geodesics, exponential and logarithm maps), a simpler approach consists in working on an affine section of the quotient. Consider an equivalence class $[Y_{A}]$ with $Y_A$ a representative of the class. We define the section of the quotient at $Y_A$ as the set of points:
\begin{equation}\label{eq:section}
  {\mathcal{S}_{A}} \coloneqq \{ {Y}_{A} \left( I + ({Y}_{A}^\T  {Y}_{A})^{-1} S  \right) + {{Y}_{{A}\perp}} K \, \vert \, S^\T = S, \ S \succ - {Y}_{A}^\T {Y}_{A}, \ K \in \R^{(n-r) \times r} \},
\end{equation}
where the matrix ${Y}_{{A}\perp}\in \R^{n \times (n-r)}$ is any orthonormal basis for the orthogonal complement of $Y_A$, i.e., $Y_{A}^\T Y_{{A}\perp}=0$ and ${Y}_{{A}\perp}^\T {Y}_{{A}\perp}=I_{n-r}$. The constraint $S \succ - {Y}_{A}^\T {Y}_{A}$ guarantees that there is at most one representative of each equivalence class $[Y_B]$ in the section, and exactly one under the generic condition that $Y_A^TY_B$ is nonsingular.

Consider the section of the quotient at $Y_{A_1}$. The representative in the section of any equivalence class $[Y_{A_2}]$ (with $Y_{A_1}^TY_{A_2}$ nonsingular) is then $\bar Y_{A_2 } = Y_{A_2} Q^\T$, where $Q$ is the orthogonal factor of the polar decomposition of $Y_{A_1}^\T Y_{A_2}$. Once all the points are projected on the section, we can simply perform Euclidean operations on the section.

%
%



\section{Construction of low-rank covariance functions}\label{sec:surfaces}

A low-rank covariance function is a mapping from a set of parameters to a low-rank PSD matrix.

\begin{definition}[\textbf{Low-rank covariance function and family}]\label{def:singular_family}
A $p$-parameter low-rank covariance function is a map ${\varphi}:\mathbb{R}^p \rightarrow \cup_{k=0}^{r}S_{+}(k,n)$, for $r<n$; its corresponding covariance family is the image of ${\varphi}$.
\end{definition}

\subsection{First-order covariance functions}\label{sec:bilinear}

In this section, we consider two possible generalizations of multilinear interpolation to manifolds. The simplest way consists in mapping all the points to a linear approximation of the manifold (here, a section of the quotient), and applying multilinear interpolation on the section. A second approach resorts to the geodesics (generalization of straight lines) on the manifold. It is interesting to notice that both reduce to the one-parameter geodesic \cref{eq:geodesic} in the one-parameter case if the reference point of the section belongs to the geodesic (cf.\ \Cref{rem:coincide}).

\subsubsection{First-order sectional covariance function}\label{sec:psectional}

The main idea is to consider a section \cref{eq:section}, projecting the data matrices to that section and performing multilinear interpolation on it. The definition below presents the sectional covariance function in the case of bilinear interpolation: we explain the steps to obtain it. A schematic representation can be seen in \Cref{fig:trees_sectional}.

\begin{definition}[\textbf{The \emph{sectional} $p$-parameter covariance function}]\label{def:sectional}
The sectional $p$-parameter covariance function is obtained as follows:
\begin{enumerate}
\item Select a member $\bar{Y}_{A_1}$ of the equivalence class $[Y_{A_1}]$.
\item Intersect the equivalence classes of $p$ data matrices $({Y}_{A_i}, i=1\dots p)$ with the section defined by $\bar{Y}_{A_1}$  to obtain:
\begin{displaymath}
\bar{Y}_{A_i}, i=1\dots p.
\end{displaymath}
\item Perform any type of Euclidean multilinear interpolation of the projected data matrices. For instance, in the bilinear case (two parameter and four data matrices):
\begin{displaymath}
\bar{Y}_{{\varphi}_{A_1 \rightarrow \dots\rightarrow A_4}}(t_1,t_2)= \bar{Y}_{A_1}(1-t_1)(1-t_2)+\bar{Y}_{A_2}(1-t_1)t_2+\bar{Y}_{A_3}t_1(1-t_2)+\bar{Y}_{A_4}t_1t_2.
\end{displaymath}
\item Multiply the factor by its transpose to obtain the full matrix:
\begin{displaymath}
{\varphi}_{A_1 \rightarrow A_2 \rightarrow A_3\rightarrow A_4}(t_1,t_2)=\bar{Y}_{{\psi}_{A_1 \rightarrow A_2 \rightarrow A_3\rightarrow A_4}}(t_1,t_2)\bar{Y}_{{\psi}_{A_1 \rightarrow A_2 \rightarrow A_3\rightarrow A_4}}^{\T}(t_1,t_2).
\end{displaymath}
\end{enumerate}
\end{definition}

\begin{figure}[h!]
\centering
  \includegraphics[scale=1]{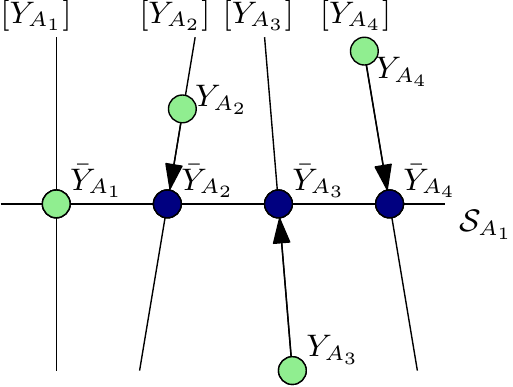}
  \caption{Schematic representation of the sectional covariance function. The vertical lines correspond to equivalence classes and the horizontal line is the section into which the data matrices are projected.}
  \label{fig:trees_sectional}
\end{figure}

\subsubsection{First-order geodesic covariance function}\label{sec:pgeodesic}

By connecting $p$ geodesics of the form in \cref{eq:geodesic}, we can obtain the \emph{geodesic} $p$-parameter covariance function, which is described below and represented in \Cref{fig:configPoints}.

%
%

\begin{definition}[\textbf{The \emph{geodesic} $p$-parameter covariance function}]\label{def:pkernelbalanced}
By connecting two one-parameter geodesics, we obtain the geodesic two-parameter covariance function:

\begin{equation}\label{eq:unbalanced}
{\varphi}_{(A_{1} \rightarrow A_{2})\rightarrow (A_{3}\rightarrow A_{4})}(t_1,t_2)\coloneqq{\varphi}_{\big({\varphi}_{A_{1} \rightarrow A_{2}}(t_1)\big)\rightarrow \big({\varphi}_{A_{3} \rightarrow A_{4}}(t_1)\big)}(t_2).
\end{equation}
Recursively, we can construct the geodesic $p$-parameter covariance function.
\end{definition}

\begin{figure}[h!]
\centering
\begin{tikzpicture}[scale = 0.6]
\draw (0,0) node[below left] {$A_1$} ;
\draw (0,0) node {$ \bullet$} ;
\draw (5,0) node[below right] {$A_2$} ;
\draw (5,0) node {$ \bullet$} ;
\draw (5,5) node[above right ] {$A_4$} ;
\draw (5,5) node {$ \bullet$} ;
\draw (0,5) node[above left] {$A_3$} ;
\draw (0,5) node {$ \bullet$} ;
\draw (2,0) node[below] {${\varphi}_{A_{1} \rightarrow A_{2}}(t_1)$} ;
\draw (2,0) node {$ \bullet$} ;
\draw (2,5) node[above] {${\varphi}_{A_{3} \rightarrow A_{4}}(t_1)$} ;
\draw (2,5) node {$ \bullet$} ;
\draw (2,3) node[right] {${\varphi}_{(A_{1} \rightarrow A_{2})\rightarrow (A_{3}\rightarrow A_{4})}(t_1,t_2)$} ;
\draw (2,3) node {$ \bullet$} ;
\draw (0,0) -- (5,0);
\draw (5,0) -- (5,5);
\draw (5,5) -- (0,5);
\draw (0,5) -- (0,0);
\draw[dotted] (2,0) -- (2,5);
\draw[myred] (0,0) -- (2,0);
\draw[myred] (0,5) -- (2,5);
\draw[mygreen] (2,0) -- (2,3);
\draw[myred] (1,0) node[above] {$t_1$} ;
\draw[myred] (1,5) node[below] {$t_1$} ;
\draw[mygreen] (2,1.5) node[left] {$t_2$} ;
\end{tikzpicture}
\label{fig:configPoints}
\caption{Configuration of the points in a Euclidean setting. Example for two parameters. }
\end{figure}
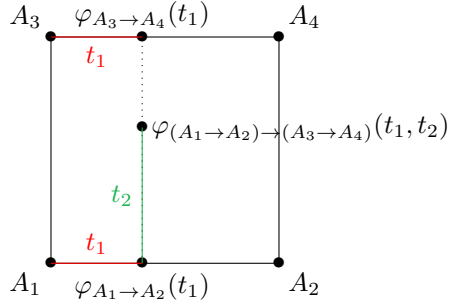

In order to build a geodesic $p$-parameter covariance function, according to the definition, we need $2^p$ data matrices.


\begin{remark}\label{rem:coincide}
For the one-parameter case, if the reference point of the section is on the geodesic, the sectional and geodesic families coincide. Indeed, using the relationships from \Cref{sec:ssd}, notice that the geodesic (geodesic one-parameter function) can be converted to a sectional one-parameter covariance function:
\begin{displaymath}
Y_{{\varphi}_{A_1 \rightarrow A_2}}(t)=Y_{A_1}+t\dot{Y}_{A_1 \rightarrow A_2}=Y_{A_1}+t(Y_{A_2}Q^\T-Y_{A_1})=Y_{A_1}(1-t)+t Y_{A_2}Q^\T
\end{displaymath}
\begin{displaymath}
=Y_{A_1}(1-t)+t\bar{Y}_{A_2},
\end{displaymath}
where $\bar{Y}_{A_2}=Y_{A_2}Q^\T$ is the projection of ${Y}_{A_2}$ into the section defined by ${Y}_{A_1}$.
\end{remark}



\subsection{Piecewise first-order covariance functions}
In practical applications, it is desirable to build covariance functions by patches. Each patch only depends on a reduced number of data matrices, usually those that are ``closest'' in some sense.
Let us illustrate using patches of two-parameter covariance functions. We assume that the data matrices used to define the family (equivalently called the ``anchor'' covariance matrices) are described by a grid of indices: this collection consists of matrices $(A_{i, j})$, with $A_{i,j} \in \psd{r}{n}$, and with $i \in \{0, \dots, N_1\}$ and $j \in \{0, \dots, N_2\}$.  The two covariance families proposed in the previous section can be computed on each patch of the grid, to build patchwise multilinear surfaces. The resulting surfaces will be denoted $\varphi^\mathrm{LS}$ and $\varphi^{\mathrm{LG}}$, where the $L$ stands for Linear (those surfaces were obtained as generalization of linear interpolation on manifolds), the $S$ for Section, and the $G$ for Geodesic.

\begin{definition}\label{def:1stOrderCovFunctions}
Given a grid of points $(A_{i, j})$, with $A_{i,j} \in \psd{r}{n}$, and with $i \in \{0, 1, \dots, N_1\}$ and $j \in \{0, 1, \dots, N_2\}$, the first-order patchwise surfaces $\varphi^{\mathrm{LS}}: [0, N_1] \times [0,N_2]  \to \psd{\leq r}{n}$ and $\varphi^{\mathrm{LG}}: [0, N_1] \times [0,N_2]  \to \psd{\leq r}{n}$ are respectively the unions of the surfaces $\varphi^{\mathrm{LS}}_{l,m}$ and $\varphi^{\mathrm{LG}}_{l,m}$, each defined as follows.
Let $l  =  0, \dots, N_1-1$ and $m = 0, \dots, N_2 - 1$ be the indices of a patch. The function $\varphi^{\mathrm{LS}}_{l,m}$, delimited by the data matrices $A_{l,m}$, $A_{l+1,m}$, $A_{l,m+1}$, $A_{l+1,m+1}$, is  the sectional two-parameter covariance function defined in~\Cref{def:sectional}, with bilinear interpolation of the representatives of the points $A_{l,m}$, $A_{l+1,m}$, $A_{l,m+1}$, $A_{l+1,m+1}$ in the section chosen for that patch. Examples of choices of sections are given in \Cref{sec:app}.

The function $\varphi^{\mathrm{LG}}_{l,m}$ is the geodesic two-parameter covariance function, defined in~\Cref{def:pkernelbalanced}, that interpolates the points $A_{l,m}$, $A_{l+1,m}$, $A_{l,m+1}$, $A_{l+1,m+1}$.
\end{definition}

\begin{proposition} \label{prop:zeroMeasure}
If for each patch, the data matrices located on the corners of the patch satisfy \Cref{hypo:noCutLocus}, it holds, except possibly for a zero measure set, that $\varphi^{\mathrm{LS}}(t_1,t_2) \in \psd{r}{n}$ and $\varphi^{\mathrm{LG}}(t_1,t_2) \in \psd{r}{n}$.
\end{proposition}

\begin{proof}
For $s \in \{LS, LG \}$, and for $(t_1,t_2) \in [0,1] \times [0,1]$, let $Y_\varphi^s(t_1,t_2)$ be a representative of the equivalence class associated with $\varphi^s(t_1,t_2)$. Consider the function $f : [0,1] \times [0,1] \to \mathbb{R}: (t_1,t_2) \mapsto \det(Y_\varphi^{s}(t_1,t_2)^\T Y_\varphi^{s}(t_1,t_2))$. This function is zero if and only if $\varphi^s(t_1,t_2)$ has rank $k < r$. Moreover, under \Cref{hypo:noCutLocus}, this function is real analytic. Indeed, in the case $s = LS$, it is a polynomial in $t_1$ and $t_2$, which is real analytic. In the case $s = LG$, it can be readily checked that:
\[ Y_\varphi^s = (1-t_1) (1-t_2) Y_{A_1} +  t_1 (1-t_2) Y_{A_2} Q_{1-2}^\T +  (1-t_1) t_2 Y_{A_3} Q(t_1)^\T + t_1 t_2 Y_{A_4} Q_{3-4}^\T Q(t_1)^\T, \]
where $Q_{1-2}$ (resp.\ $Q_{3-4}$) is the orthogonal factor of the polar decomposition of $Y_{A_1}^\T Y_{A_2}$ (resp.\ $Y_{A_3}^\T Y_{A_4}$), and  $Q(t_1)$ is the orthogonal factor of the polar decomposition of the matrix $M \coloneqq ((1-t_1) Y_{A_1} + t_1 Y_{A_2}Q_{1-2}^\T)^\T ((1-t_1) Y_{A_3} + t_1 Y_{A_4}Q_{3-4}^\T)$; see~\cite{massart2018quotient} for more information.
If \Cref{hypo:noCutLocus} holds, this matrix is nonsingular for all $t_1 \in [0,1]$~\cite{massart2018quotient}. The  polar  decomposition of a one-variable matrix function $t_1 \mapsto M(t_1)$, with $M(t_1)$ nonsingular, is real analytic~\cite[Theorem 3]{mehrmann1993numerical}. We then conclude the proof using the fact that the set of zeros of a real analytic function $f: \mathbb{R}^n \to \mathbb{R}$ has measure zero; see, e.g.,~\cite[Remark 5.23]{kuchment2016overview}.
\end{proof}

\subsection{Higher-order covariance functions using B\'ezier curves}
The previous section presented two families defined as generalizations of multilinear interpolation on the manifold. In this section, we consider higher-order interpolation on the manifold. We focus on patchwise B\'ezier surfaces on the grid.
Again, we will distinguish two cases: methods resorting to Euclidean algorithms in a section of the manifold and methods based on successive evaluations of geodesics.

\subsubsection{Higher-order sectional covariance function}\label{sec:bsectional}
This method consists in first projecting all the data matrices on a section of the manifold (cf.~\Cref{eq:section}), and then using a classical Euclidean B\'ezier interpolation algorithm in the section.

 We focus on patchwise cubic B\'ezier surfaces. Those surfaces are defined by a set of control points $(b_{i,j})_{i, j = 0, \dots, 3}$, with $b_{i,j} \in \R^{n \times r} \ \forall i,j$. The cubic B\'ezier surface is defined as:

\begin{align*}
   \beta : [0,1] \times [0,1] &\to \R^{n \times r} : \\
    (t_1, t_2) \mapsto & \beta(t_1, t_2;(b_{i,j})_{i, j = 0, \dots, 3}) \coloneqq \sum_{i = 0}^3 \sum_{j = 0}^3 b_{i,j} B_{i,3}(t_1) B_{j,3}(t_2),
\end{align*}
where $B_{i,3}(t) \coloneqq  \sum_{i = 0}^3 \begin{pmatrix}  3 \\ i \end{pmatrix} t^i (1-t)^{3-i},$ with $t \in [0,1],$ is the Bernstein polynomial of order $3$.

We define the patchwise cubic B\'ezier surface on the section, denoted  ${\varphi}^{\mathrm{BS}}$, as follows.

\begin{definition} \label{def:BezierCovFunctionSection}
Assume that we have a set of data matrices $(A_{i,j})$, with $A_{i,j} \in \psd{r}{n}$, and with $i = 0, \dots, N_1$ and $j = 0, \dots, N_2$. Choose a section of the quotient manifold and map all the data matrices to this section to obtain $(Y_{i,j})$, with $Y_{i,j} \in \R_*^{n \times r}$, $i = 0, \dots, N_1$ and $j = 0, \dots, N_2$. Let $Y_{\varphi^{\mathrm{BS}}}: [0,N_1] \times [0,N_2] \to \R^{n \times r}$ be the patchwise B\'ezier surface interpolating the data matrices $(Y_{i,j}),$ and computed as in~\cite{absil2016differentiable}. The surface ${\varphi}^{\mathrm{BS}}$ is obtained as follows.
\begin{equation}
    \varphi^{\mathrm{BS}}: [0,N_1] \times [0,N_2] \to \psd{\leq r}{n} :  (t_1, t_2) \mapsto Y_{\varphi^{\mathrm{BS}}}(t_1,t_2) Y_{\varphi^{\mathrm{BS}}}(t_1,t_2)^\T.  \label{eq:beziersection}
\end{equation}
\end{definition}

\begin{proposition}
 Apart from possibly a zero measure set, it holds that $\varphi^{\mathrm{BS}} \in \psd{r}{n}$.
\end{proposition}

\begin{proof}
For all $(t_1, t_2) \in  [0,1] \times [0,1]$, let $Y_\varphi^{\mathrm{BS}}(t_1,t_2)$ be a representative in $\mathbb{R}_*^{n \times p}$ of $\varphi^{\mathrm{BS}}(t_1,t_2)$. The proof is an immediate extension of~\Cref{prop:zeroMeasure}, using the fact that  $Y_\varphi^{\mathrm{BS}}(t_1,t_2)$ is a polynomial function of $t_1$ and $t_2$.
\end{proof}
\subsubsection{Higher-order covariance function based on the exp and log}\label{sec:bgeodesic}
In this case, we refer to the B\'ezier surface interpolation algorithm on manifolds, proposed in~\cite{absil2016differentiable}. This algorithm relies on the expressions for the logarithm and the exponential map, provided in~\cite{massart2018quotient} for the manifold $\psd{r}{n}$. We define the surface   ${\varphi}^{\mathrm{BG}}$ (patchwise B\'ezier-like on the manifold) as follows.
\begin{definition} \label{def:BezierCovFunctionManifold}
 Let $(A_{i,j})$ be a data set, with $A_{i,j} \in \psd{r}{n}$, and with $i = 0, \dots, N_1$, $j = 0, \dots, N_2$.  The surface ${\varphi}^{\mathrm{BG}}:[0, N_1] \times [0,N_2] \to \psd{\leq r}{n}$ is the surface obtained by applying the B\'ezier surface interpolation algorithm on manifolds proposed in~\cite{absil2016differentiable}, with type-II reconstruction (cf.~\cite[Definition 4]{absil2016differentiable}), and with the control points chosen as suggested in~\cite{Absil2016a}.
\end{definition}

\begin{proposition}
Under~\Cref{hypo:noCutLocus}, it holds that $\varphi^{\mathrm{BG}} \in \psd{r}{n}$, apart from possibly a zero measure set.
\end{proposition}

\begin{proof}
For all $(t_1, t_2) \in  [0,1] \times [0,1]$, let $Y_\varphi^{\mathrm{BG}}(t_1,t_2)$ be a representative in $\mathbb{R}_*^{n \times p}$ of $\varphi^{\mathrm{BG}}(t_1,t_2)$. The proof is an immediate extension of~\Cref{prop:zeroMeasure}, using the fact that  $Y_\varphi^{\mathrm{BG}}(t_1,t_2)$ can be written only in terms of polynomial terms in the variables $t_1$ and $t_2$, and orthogonal factors of some real analytic matrix functions $t_1 \mapsto M(t_1)$. (This last fact is a direct consequence of the reconstruction method used: interpolation is first performed along one variable, and then along the other.)
\end{proof}

\subsection{Interpolation of labelled matrices using covariance functions}\label{sec:interpol}
If the data matrices $(A_i)$ are labelled by certain known parameters $x$, i.e., $A_i = A(x_i)$, it is also possible to use the covariance functions in \Cref{sec:surfaces} to perform interpolation. (In  \Cref{sec:app}, for example, $x$ will correspond to the magnitude and heading of the prevailing wind.)

\begin{problem}[{\textbf{Interpolation with low-rank covariance functions}}]\label{prob:interpol}
Given matrices $A(x_i) \in \psd{r}{n}$ for $x_0,x_1,\dots,x_m\in\mathbb{R}^p$ and an associated low-rank covariance function, evaluate this function for any $x\in\mathbb{R}^p$.
 \end{problem}

Of course, performing interpolation requires mapping from the label $x \in \mathbb{R}^p$ to an element of the input domain of the chosen low-rank covariance function. For example, when using the two-parameter covariance functions detailed above, constructed from $N_1 N_2$ data matrices, we need to map from a subset of $\mathbb{R}^2$ to $[0, N_1] \times [0, N_2]$. We will usually use  affine mappings for this purpose; more details are given in \Cref{sec:app}, where we evaluate the interpolation capabilities of one-parameter (\Cref{sec:cs_interpol}) and multi-parameter (\Cref{sec:interpol_2p}) covariance functions.




\section{Covariance identification using distance minimization}\label{sec:results}
Having proposed multi-parameter low-rank covariance families in the previous section, we can now describe identification procedures within such families. That is, given a data set $(y_i)_{i=1}^q$ (assumed to be centered, with $y_i \in \mathbb{R}^n$) from which we construct a sample covariance matrix $\widehat{C} = \frac{1}{q}\sum_i y_i y_i^\top$,
we would like to find the most representative member of a family.

A widely used methodology for selecting a representative member of any parametric covariance family (like ${\varphi}_{A_{1}, \dots, A_{m}}(t)$) is maximum likelihood estimation. That is, the data are assumed to have a certain probability distribution, e.g., $Y_i\sim N\big{(}0,{\varphi}_{A_{1}, \dots, A_{m}}(t)\big{)}$, and we choose $t\in \mathbb{R}^{p}$ to maximize the resulting likelihood function $p(y_1, \ldots, y_q \vert t)$. Since ${\varphi}_{A_{1}, \dots, A_{m}}(t)$ are low-rank, however, the Gaussian distribution of the data is degenerate. Thus the associated probability density function is not well-defined for generic $y_i$ and maximizing the likelihood becomes non-trivial. (If the matrices were instead SPD and as $q \to \infty$, this problem is equivalent to minimizing the Kullback--Leibler divergence of $N(0, \widehat{C})$ from $N\big{(}0,{\varphi}_{A_{1}, \dots, A_{m}}(t)\big{)}$, known as reverse information projection \cite{gallager1968information,cover2012elements,csiszar2004information}.)
%

Rather than maximizing the likelihood, we propose to minimize a particular distance from the covariance function to the sample covariance matrix $\widehat{C}$.
\begin{problem}[\textbf{Minimum Frobenius distance covariance identification}]\label{prob:minsingular}
\begin{displaymath}
{\arg\min}_{t\in \mathbb{R}^{p}}\;d_F\big{(}{\varphi}_{A_{1}, \dots, A_{m}}(t),\widehat{C}\big{)},\end{displaymath}
where $d_F(A_{1},A_{2})=\| A_{1}-A_{2} \|_F$ is the Frobenius distance.
\end{problem}
This is a particular instance of ``minimum distance estimation''; such estimators, in general, have a long history in statistics \cite{wolfowitz1957minimum}.


We now discuss solutions to \Cref{prob:minsingular} given particular constructions of the covariance function $\varphi(t)$.
Since the geodesic and sectional approaches to define covariance functions coincide for the one-parameter case under some conditions (cf.\ \Cref{rem:coincide}), we divide this section in two parts: the one-parameter case ($p = 1$) and the multi-parameter case. For the latter, we focus on the case of two parameters ($p = 2$).

\subsection{One-parameter first-order covariance function}

For $p=1$, the covariance function $t_1 \mapsto \varphi(t_1)$ is simply the geodesic between the two data matrices. The optimization problem has a closed form solution that is presented below.

\begin{proposition}[\textbf{Solution of the low-rank covariance identification problem}]\label{res:singularsol}
The solutions of \Cref{prob:minsingular} for $p=1$ are the roots of a third order polynomial $at^3+bt^2+ct+d=0$ with:
\begin{displaymath}
a=4 \Tr(\dot{Y}_{A_1 \rightarrow A_2}\dot{Y}_{A_1 \rightarrow A_2}^{\T}\dot{Y}_{A_1 \rightarrow A_2}\dot{Y}_{A_1 \rightarrow A_2}^{\T}),
\end{displaymath}
\begin{displaymath}
b=12\Tr(Y_{A_1}\dot{Y}_{A_1 \rightarrow A_2}^{\T}\dot{Y}_{A_1 \rightarrow A_2}\dot{Y}_{A_1 \rightarrow A_2}^{\T}),
\end{displaymath}
\begin{displaymath}
c=4\Tr(2Y_{A_1}Y_{A_1}^{\T}\dot{Y}_{A_1 \rightarrow A_2}\dot{Y}_{A_1 \rightarrow A_2}^{\T}+Y_{A_1}\dot{Y}_{A_1 \rightarrow A_2}^{\T}Y_{A_1}\dot{Y}_{A_1 \rightarrow A_2}^{\T}-\dot{Y}_{A_1 \rightarrow A_2}\dot{Y}_{A_1 \rightarrow A_2}^{\T}Y_{\widehat{C}}Y_{\widehat{C}}^{\T}),
\end{displaymath}
\begin{displaymath}
d=4\Tr(Y_{A_1}Y_{A_1}^{\T}\dot{Y}_{A_1 \rightarrow A_2}Y_{A_1}^{\T}-\dot{Y}_{A_1 \rightarrow A_2}Y_{A_1}^{\T}Y_{\widehat{C}}Y_{\widehat{C}}^{\T}).
\end{displaymath}

\end{proposition}
\begin{proof}
The cost function is:
\begin{displaymath}
d_F\big{(}{\varphi}_{A_{1} \rightarrow A_{2}}(t),\widehat{C}\big{)}=\sqrt{\Tr\big{(}({\varphi}_{A_{1} \rightarrow A_{2}}(t)-\widehat{C})({\varphi}_{A_{1} \rightarrow A_{2}}(t)-\widehat{C})^{\T}\big{)}}.
\end{displaymath}
The third-order polynomial is obtained after setting the derivative to zero and noting that the optimization problem is unconstrained.
\end{proof}

\Cref{res:singularsol} provides at least one solution. If there are three roots, the minimizer is of course the one with smallest objective. As with any third order polynomial, the uniqueness condition of this solution is:
\begin{displaymath}
18abcd-4b^{3}d+b^{2}c^{2}-4ac^{3}-27a^{2}d^{2}\leq0.
\end{displaymath}

The computational cost of finding the solution of the low-rank covariance identification problem is only $\mathcal{O}(nr)$. Indeed, roots of the cubic equation have a closed form expression whose evaluation does not require any meaningful cost. The only computational cost is that associated with computing traces to obtain the polynomial coefficients. By virtue of the cyclic property of the trace, we can compute these traces with $\mathcal{O}(nr)$ elementary operations.

%

\subsection{Two-parameter first-order covariance functions}
Here we focus on~\Cref{prob:minsingular} in the two-parameter case ($p=2$) for first-order covariance functions, i.e.,  $\varphi = \varphi^{\mathrm{LS}}$ or $\varphi = \varphi^{\mathrm{LG}}$ (cf.~\Cref{def:1stOrderCovFunctions}).
%
Similarly to the previous sections, we assume that data matrices are defined on a grid of points $(A_{i,j})$, with $A_{i,j} \in \psd{r,n}$, $i = 0, \dots, N_1$ and $j = 0, \dots, N_2$.

\subsubsection{First-order sectional covariance function} \label{sec:optiLinSec}
In the case $\varphi = \varphi^{\mathrm{LS}}$ (see~\Cref{def:1stOrderCovFunctions}), we propose to use a gradient descent on each patch of the surface. Observe indeed that the surface $\varphi^{\mathrm{LS}}$ is generally nondifferentiable (actually, even noncontinuous) on the borders of the patches. The global optimum is then computed as the minimum of the optima obtained on the patches.
Let
\begin{equation}
    f : [0, N_1] \times [0,N_2] \to \R : (t_1,t_2) \mapsto f(t_1,t_2)\coloneqq  d_F(\varphi^{\mathrm{LS}}(t_1,t_2), \widehat{C})^2
    \label{eq:costFunctionLS}
\end{equation}
be the cost function to minimize.

Consider an arbitrary patch $(l,m)$, with  $l = 0, \dots, N_1 - 1$ and $m = 0, \dots, N_2 - 1$. Let $f^{l,m}$ be the restriction of $f$ to the patch $(l,m)$, and let $A_{l,m}$, $A_{l+1,m}$, $A_{l,m+1}$, $A_{l+1,m+1}$ denote the four corners of the patch $(l,m)$, and $\bar Y_{A_{l,m}}$, $\bar Y_{A_{l+1,m}}$, $\bar Y_{A_{l,m+1}}$, $\bar Y_{A_{l+1,m+1}}$ their projection on the section. We omit the superscript ${\mathrm{LS}}$ in the remainder of this section. The gradient of the restriction of the cost function~\eqref{eq:costFunctionLS} to that patch can be computed explicitly:
\begin{align}
\frac{\partial f^{l,m}}{\partial t_1}(t_1,t_2)= 2 \Tr{D \varphi(t_1, t_2) [e_1] (\varphi(t_1, t_2) - \widehat{C})^\T  } ,\\
\frac{\partial f^{l,m}}{ \partial t_2 }(t_1,t_2) = 2 \Tr{D \varphi(t_1, t_2) [e_2] (\varphi(t_1, t_2) - \widehat{C})^\T  },
\end{align}
with $D \varphi(t_1, t_2) [e_1]$ and $D \varphi(t_1, t_2) [e_2]$ defined as:
\begin{align*}
D \varphi(t_1, t_2) [e_1] &=  DY_{\varphi}(t_1, t_2)[e_1] Y_{\varphi}^\T + Y_{ \varphi} DY_{\varphi}(t_1, t_2)[e_1]^\T, \\
D \varphi(t_1, t_2) [e_2] &=  DY_{\varphi}(t_1, t_2)[e_2] Y_{\varphi}^\T + Y_{ \varphi} DY_{\varphi}(t_1, t_2)[e_2]^\T.
\end{align*}
The factor $Y_{\varphi}$ was defined in~\Cref{def:sectional}:
\[ Y_{\varphi} = \bar{Y}_{A_{l,m}}(1-t_1)(1-t_2)+\bar{Y}_{A_{l+1,m}}(1-t_1)t_2+\bar{Y}_{A_{l,m+1}}t_1(1-t_2)+\bar{Y}_{A_{l+1,m+1}}t_1t_2. \]
So $DY_{\varphi}(t_1, t_2)[e_1]$ and $DY_{\varphi}(t_1, t_2)[e_2]$ can be easily obtained as:
\begin{align*}
DY_{\varphi}(t_1, t_2)[e_1] &= t_2 (\bar Y_{A_{l,m}} - \bar Y_{A_{l+1,m}} - \bar Y_{A_{l,m+1}} + \bar Y_{A_{l+1,m+1}}) + \bar Y_{A_{l,m+1}} - \bar Y_{A_{l,m}}, \\
DY_{\varphi}(t_1, t_2)[e_2] &= t_1 (\bar Y_{A_{l,m}} - \bar Y_{A_{l+1,m}} - \bar Y_{A_{l,m+1}} + \bar Y_{A_{l+1,m+1}}) + \bar Y_{A_{l+1,m}} - \bar Y_{A_{l,m}},
\end{align*}
with the only difference being the parameter $t_1$ vs.\ $t_2$.

\subsubsection{First-order geodesic covariance function}\label{sec:variablemin}


We focus now on \Cref{prob:minsingular} for $p=2$, when the covariance function is the surface $\varphi^{\mathrm{LG}}$ defined in~\Cref{def:1stOrderCovFunctions}. Let
\begin{equation}
    f : [0, N_1] \times [0,N_2] \to \R : (t_1,t_2) \mapsto f(t_1,t_2)\coloneqq  d_F(\varphi^{\mathrm{LG}}(t_1,t_2), \widehat{C})^2
    \label{eq:costFunctionLG}
\end{equation}
be the cost function. The surface  $\varphi^{\mathrm{LG}}$ is generally not differentiable on the borders of the patches. As a result, similarly to the previous section, we propose to run an optimization algorithm to find the optimum on each patch, and to compare the optimal values obtained on the patches to obtain the global optimum.

Let $f^{l,m}$ be the restriction of $f$ to the patch $(l,m)$. We propose to minimize $f^{l,m}$ by expressing it as a one-variable function, replacing $t_2$ by its optimal value:
\begin{equation}
t_2^*(t_1) = \argmin_{t_2 \in \R} f^{l,m}(t_1,t_2),
\label{eq:t2opti}
\end{equation}
and then to apply gradient descent to the problem:
\begin{equation}
\min_{t_1 \in \R} \tilde f^{l,m}(t_1) \coloneqq f^{l,m}(t_1,t_2^*(t_1)).
\end{equation}
The computation of the partial derivatives required by both steps is deferred to \Cref{ap:variableprojection}.

\subsection{Higher-order covariance functions using B\'ezier curves}
We now solve~\Cref{prob:minsingular} for $p=2$ when the surface is defined from B\'ezier interpolating surfaces.
\subsubsection{Higher-order  sectional  covariance  function}
To solve~\Cref{prob:minsingular} for $p=2$ when the surface is a Euclidean B\'ezier surface built in a given section of the manifold, we propose again to use steepest descent. The cost function
\begin{equation}
    f : [0, N_1] \times [0,N_2] \to \R : (t_1, t_2) \mapsto f(t_1,t_2) \coloneqq d_F(\varphi^{\mathrm{BS}}(t_1,t_2),\widehat{C})^2,
\end{equation}
with $\varphi^{\mathrm{BS}}$ defined in~\Cref{def:BezierCovFunctionSection}, is $\mathcal{C}^1$. Moreover, since B\'ezier curves in the Euclidean space are weighted sums of Bernstein polynomials, the gradient can be computed explicitly. The computation of the gradient is similar to~\Cref{sec:optiLinSec}, except that now $Y_{\varphi}$ is obtained as a linear combination of cubic Bernstein polynomials:
\[ Y_{\varphi}(t_1,t_2) = \sum_{i = 0}^3 \sum_{j = 0}^3 b_{ij} B_{i3}(t_1) B_{j3}(t_2).   \]
The derivatives $DY_{\varphi}(t_1, t_2)[e_1]$ and $DY_{\varphi}(t_1, t_2)[e_2]$ become:
\begin{align*}
DY_{\varphi}(t_1,t_2)[e_1] &= \sum_{i = 0}^3 \sum_{j = 0}^3 b_{ij}  \dot{B}_{i3}(t_1) B_{j3}(t_2), \\
DY_{\varphi}(t_1,t_2)[e_2] &= \sum_{i = 0}^3 \sum_{j = 0}^3 b_{ij} B_{i3}(t_1)  \dot{B}_{j3}(t_2).
\end{align*}

\subsubsection{Higher-order covariance function based on the exponential and logarithm maps}

For $\varphi^{\mathrm{BG}}$ in~\Cref{def:BezierCovFunctionManifold}, it remains unclear whether the gradient of the cost function:
\begin{equation}
    f : [0, N_1] \times [0,N_2] \to \R : (t_1,t_2) \mapsto d_F(\varphi^{\mathrm{BG}}(t_1,t_2), \widehat{C})^2.
\end{equation}
has an analytical expression. Variable projection methods also do not seem applicable in this case. Thus we have to estimate the gradient numerically, resorting to finite differences.

\section{Case study: wind field approximation}\label{sec:app}

Given the increasing popularity of unmanned aerial vehicles (UAVs) in transportation, surveillance, agriculture, and beyond, accurate and safe aerial navigation is essential.
Achieving these requirements demands expressive models of the UAV's environment---in particular, the wind field---and the ability to update these models given new observations, e.g., via Kalman filtering \cite{palanthandalam2008wind,doekemeijer2017ensemble}. To this end, we wish to construct and estimate the covariance of spatially distributed wind velocity components.

\subsection{Model problem and data set}

Gaussian random field (GRF) models have previously been used to describe wind velocities (e.g., \cite{yang2017real,lawrance2010simultaneous}). A common practice in this setting is to define the covariance matrix of the velocities using the (smooth) squared-exponential kernel, perhaps with some modifications to allow for non-stationarity \cite{lawrance2011path}.
We instead assume to have instances of the covariance matrix for different values of the prevailing wind heading $\theta$ and magnitude $W$; from these instances, we will build a covariance family for continuous $(\theta,W)$. The wind field can change dramatically as function of the prevailing wind, and thus it is useful to consider a covariance family built from a variety of representative prevailing wind settings.

In general, these instances could be estimated from observational data, or they could be constructed using offline (and potentially expensive) computational fluid dynamics simulations. Here we use the latter: we solve the unsteady incompressible Navier--Stokes equations on the two-dimensional domain shown in \Cref{fig:snaps}, using direct numerical simulation with a spectral element method. The Reynolds number in our simulations, defined according to the side-length of the central obstacle, is around 500 for $W=7.0$. For each chosen value of $(\theta, W)$, we run the simulation until any transients due to the initial condition have dissipated and then collect instantaneous velocity fields as ``snapshots,'' shown in \Cref{fig:snaps}. The sample covariance of these snapshots provides the data covariance matrix at that $(\theta, W)$.

The right plot of \Cref{fig:casestudy2} represents a notional idea of our example domain: flow around a rectangular cuboid in three dimensions. We consider only a horizontal ``slice'' of this domain, e.g., the wind in a plane at height $h$ sufficiently far from the ground and from top of the obstacle so that a two-dimensional approximation is reasonable. The left plot of \Cref{fig:casestudy2} shows the mean value of the velocity on this plane, at an example value of $(\theta, W)$. The grid size is 39 $\times$ 39, and hence the discretized wind field has $n=3024 = 2 \times (39^2 - 9)$ degrees of freedom: two velocity components at each grid point, subtracting 9 points for the obstacle.  The grey contours represent the pointwise variance of the $x$-velocity plus that of the $y$-velocity (i.e., the sum of two diagonal entries of the covariance matrix, at each point in space). Naturally, the variance is larger downstream of the obstacle, where vortices are shed.

Our data set for the examples below comprises a set of covariance matrices $C(\theta_k,W_i)$, with $\theta_k=(k-1)\pi/64$, $k\in{1,\dots,32}$ and $W\in \{4.0,5.5,7.0,8.5,10.0,11.5,13.0\}$, as illustrated in \Cref{fig:data}. Using a truncated singular value decomposition of each matrix, we reduce the rank to $r=20$. These covariance matrices then belong to $\psd{20}{3024}$.

\begin{figure}[h!]
\centering
  \includegraphics[scale=0.48]{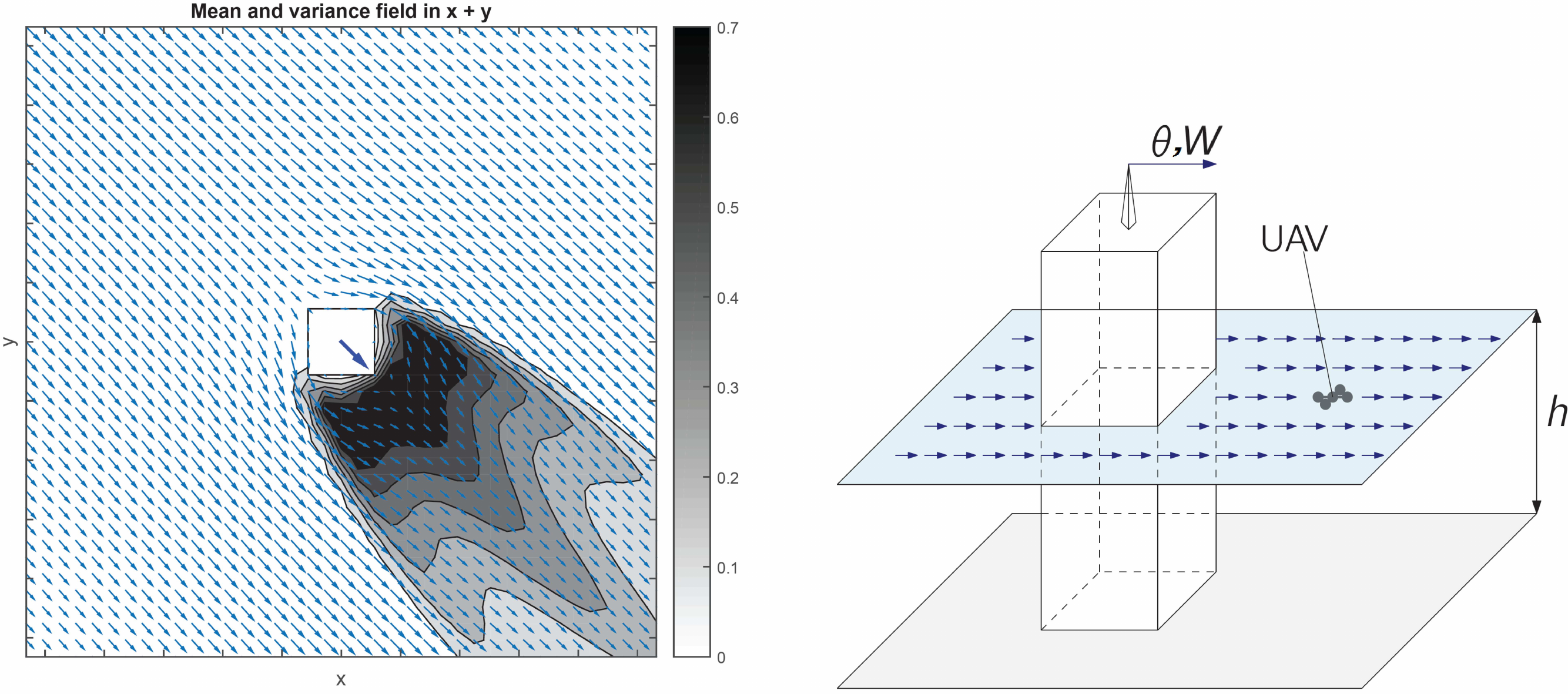}
  \caption{Representation of the wind field. Left: two-dimensional domain, with wind field around the square obstacle represented by light blue arrows and the prevailing wind in dark blue; gray contours are the variance field. Right: notional 3-D problem, with a section of the wind field at an altitude $h$.}
  \label{fig:casestudy2}
\end{figure}

\begin{figure}[h!]
\begin{multicols}{2}
\includegraphics[scale=0.47]{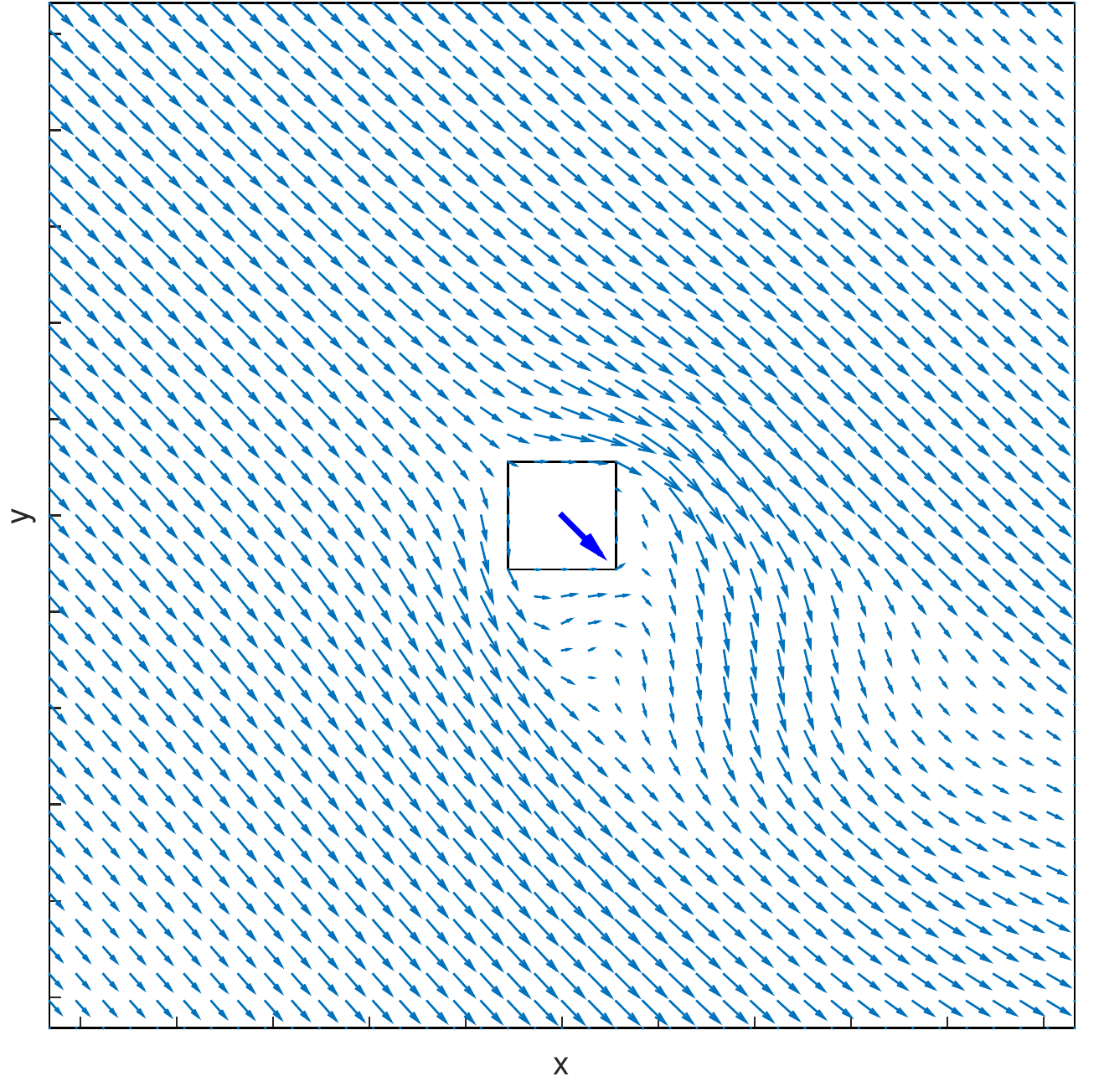}\par
\includegraphics[scale=0.47]{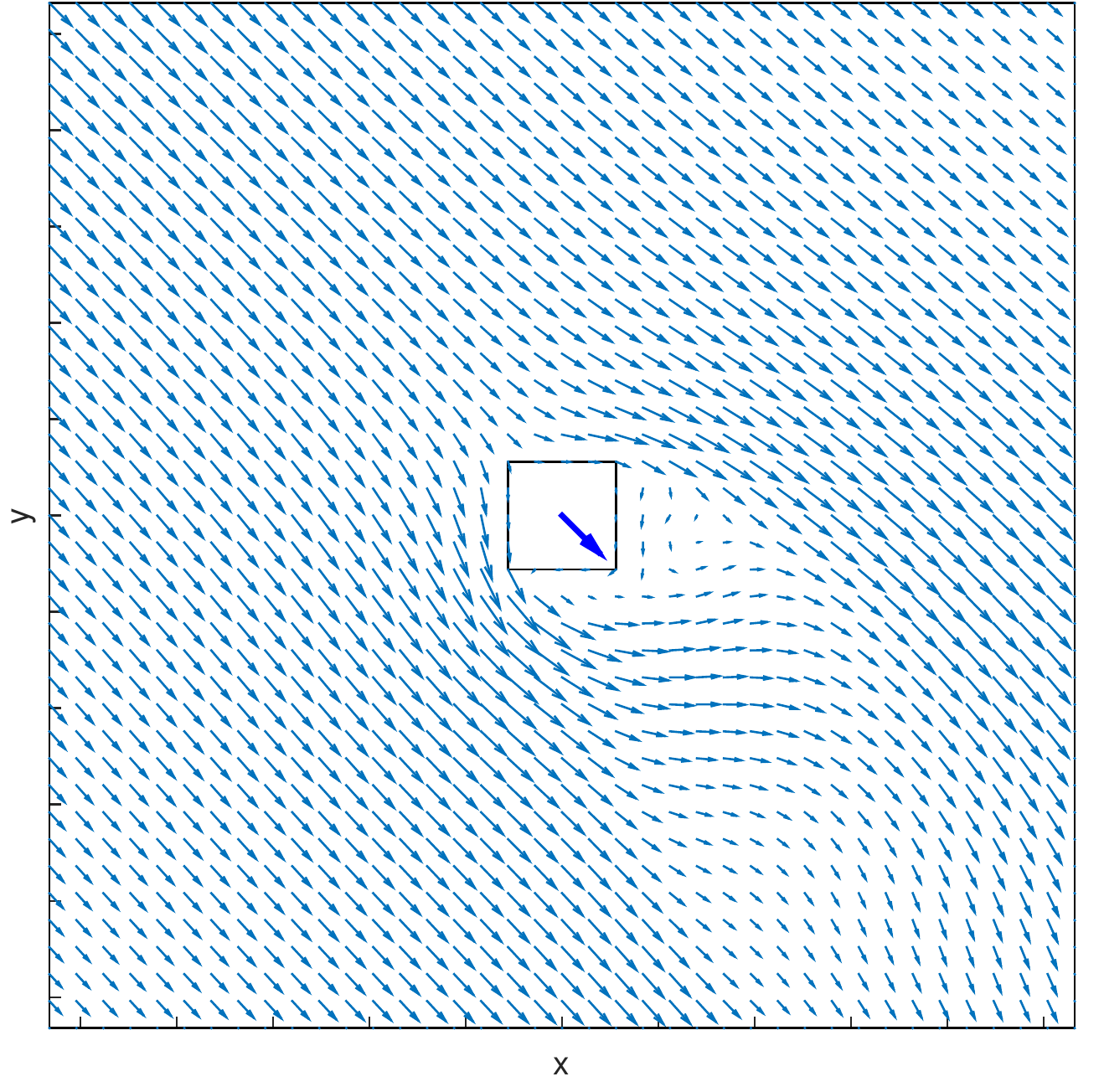}\par
    \end{multicols}
\begin{multicols}{2}
\includegraphics[scale=0.47]{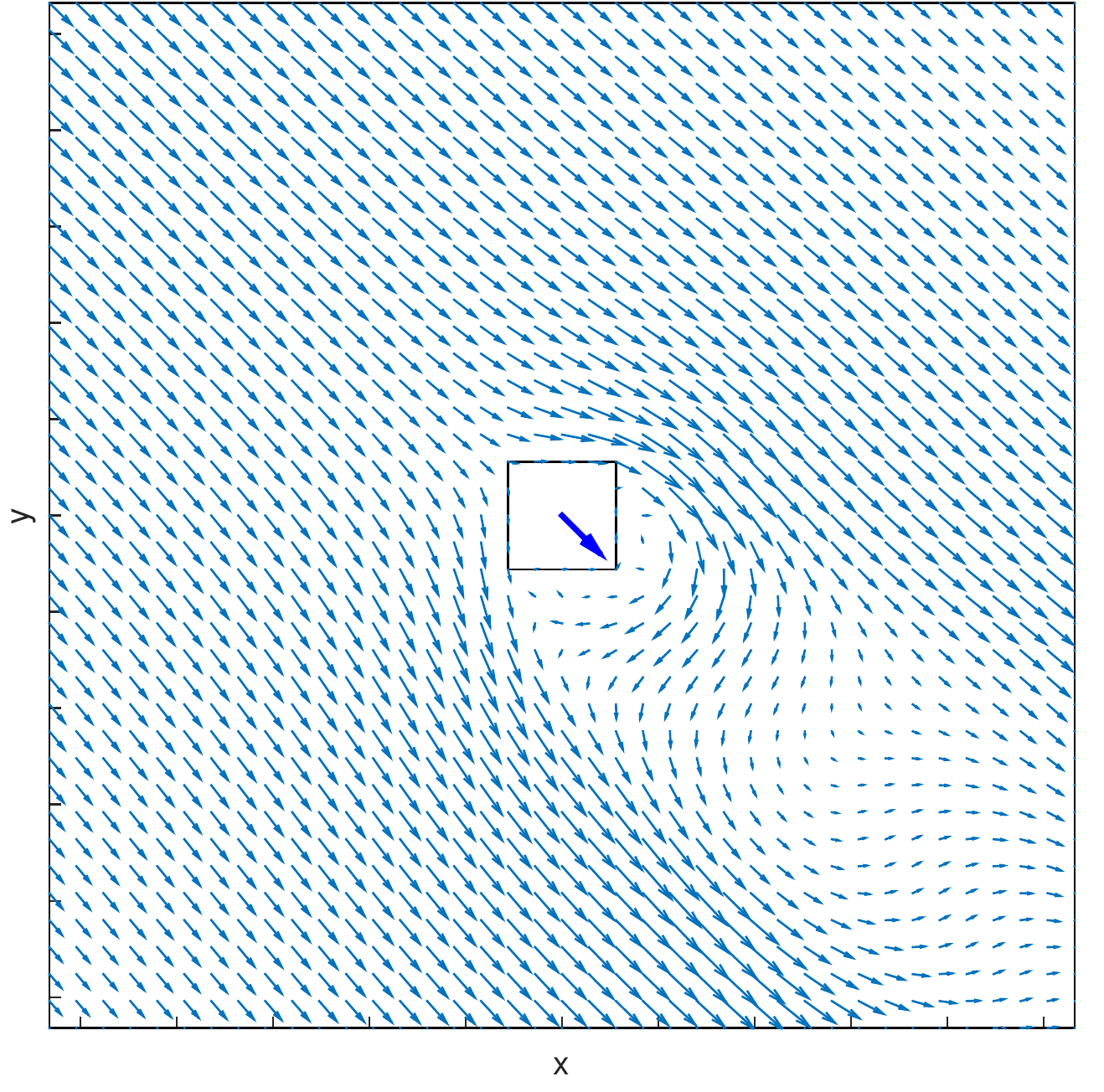}\par
\includegraphics[scale=0.47]{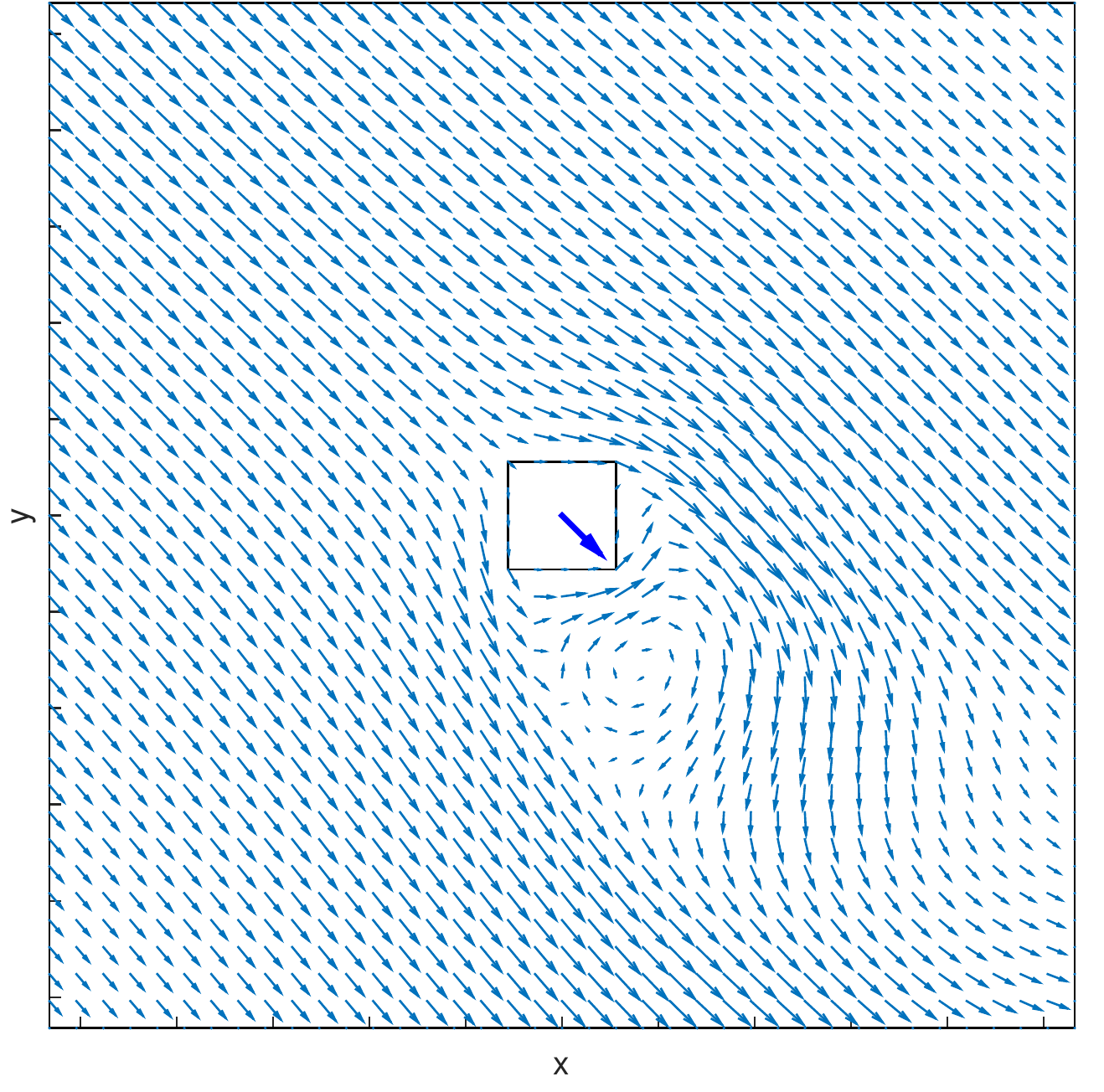}\par
\end{multicols}
  \caption{Instantaneous snapshots of the wind velocity field for $\theta=45$ and $W=7$.}
  \label{fig:snaps}
\end{figure}

\subsection{One-parameter covariance families}
\label{sec:cs_interpol}

We first consider interpolation and identification with a one-parameter geodesic covariance function $\varphi_{A_{1} \rightarrow A_{2}}(t)$, where the data matrices $A_1$ and $A_2$ are obtained at the same wind magnitude but at faraway headings: $A_1 = C(\theta_{1}=0, W=8.5)$ and $A_2 = C(\theta_{9}=23, W=8.5)$. (As noted in \Cref{rem:coincide}, the one-parameter geodesic and sectional families coincide.) To understand the relationship between the geodesic parameter $t$ and the true wind heading, we minimize the distance from each of the seven intermediate data matrices $C(\theta_{k}, 8.5)$, $k\in{1,2,\dots,9}$, to this covariance family (cf.\ red line in \Cref{fig:data}) and obtain a value of $t_k$. \Cref{fig:tvstheta} shows the resulting pairs $(\theta_k,t_k)$. The relationship between $t$ and $\theta$ is monotone and nearly linear. Similar results can be obtained for other choices of $W$.

\begin{figure}[h!]
\centering
  \includegraphics[scale=0.55]{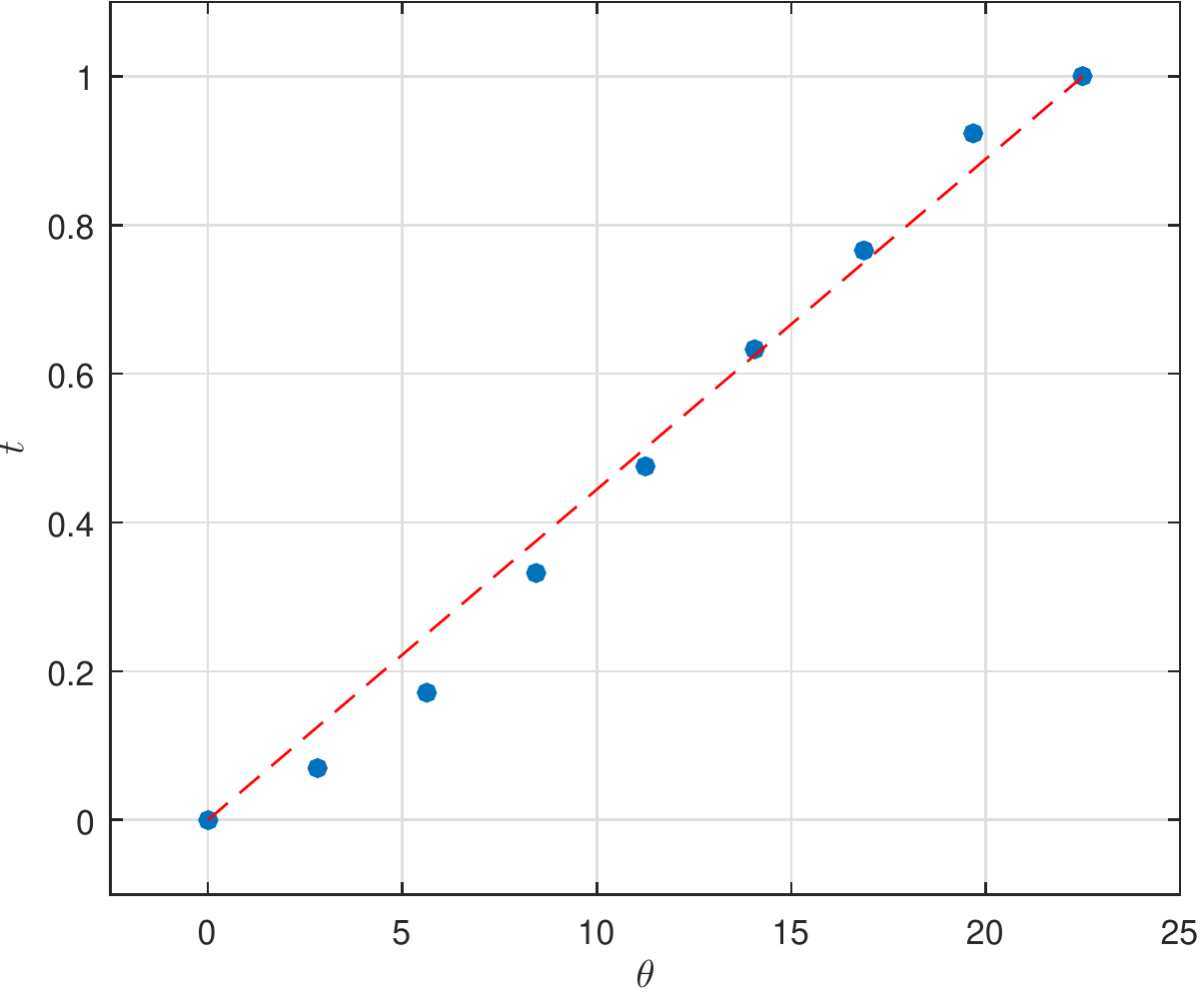}
  \caption{Minimizing value of $t$ (blue points) for data drawn from a range of wind headings $\theta$, for $W=8.5$ (cf.\ \Cref{sec:cs_interpol}). The red line represents a ``perfect'' linear relationship.}
  \label{fig:tvstheta}
\end{figure}

Next we focus on the shape of the objective function used in distance minimization for one-dimensional covariance families. We build a one-parameter covariance function $\varphi_{A_{1} \rightarrow A_{2}}(t)$ with $A_1=C(16.9,7)$ and $A_2=C(22.5,7)$ and evaluate the distance to $A_3=C(19.7,7)$ as a function of $t \in [0,1]$. (See \Cref{fig:data}, dashed blue line, to identify the relevant matrices in our data set.) This exercise is shown in \Cref{fig:casestudy}, where the anchor or data matrices $A_1, A_2$ are illustrated via inset plots with a green obstacle. (The matrices are visualized by their variance fields, as in \Cref{fig:casestudy2} (left).) First, we note that the distance objective is smooth and convex (on $[0,1]$), and that its minimum (marked with a blue dot) is close to, though not exactly, $t=0.5$. This offset is a further instance of the difference illustrated in \Cref{fig:tvstheta}, between the minimum-distance points and a perfect linear relationship. The inset plots in \Cref{fig:casestudy} with white obstacles show covariances in the one-dimensional family at intermediate values of $t$; we see that these covariances look physically reasonable, suggesting intermediate wind headings as desired. Nonetheless, we also note that the minimum Frobenius distance from $A_3$ to this family is roughly 14, about half of the distance from $A_3$ to the anchor $A_1$. For a more accurate representation of $A_3$, one may thus want a richer family or more representative data matrices. We will explore these choices below.

\begin{figure}[h!]
\centering
  \includegraphics[scale=0.48]{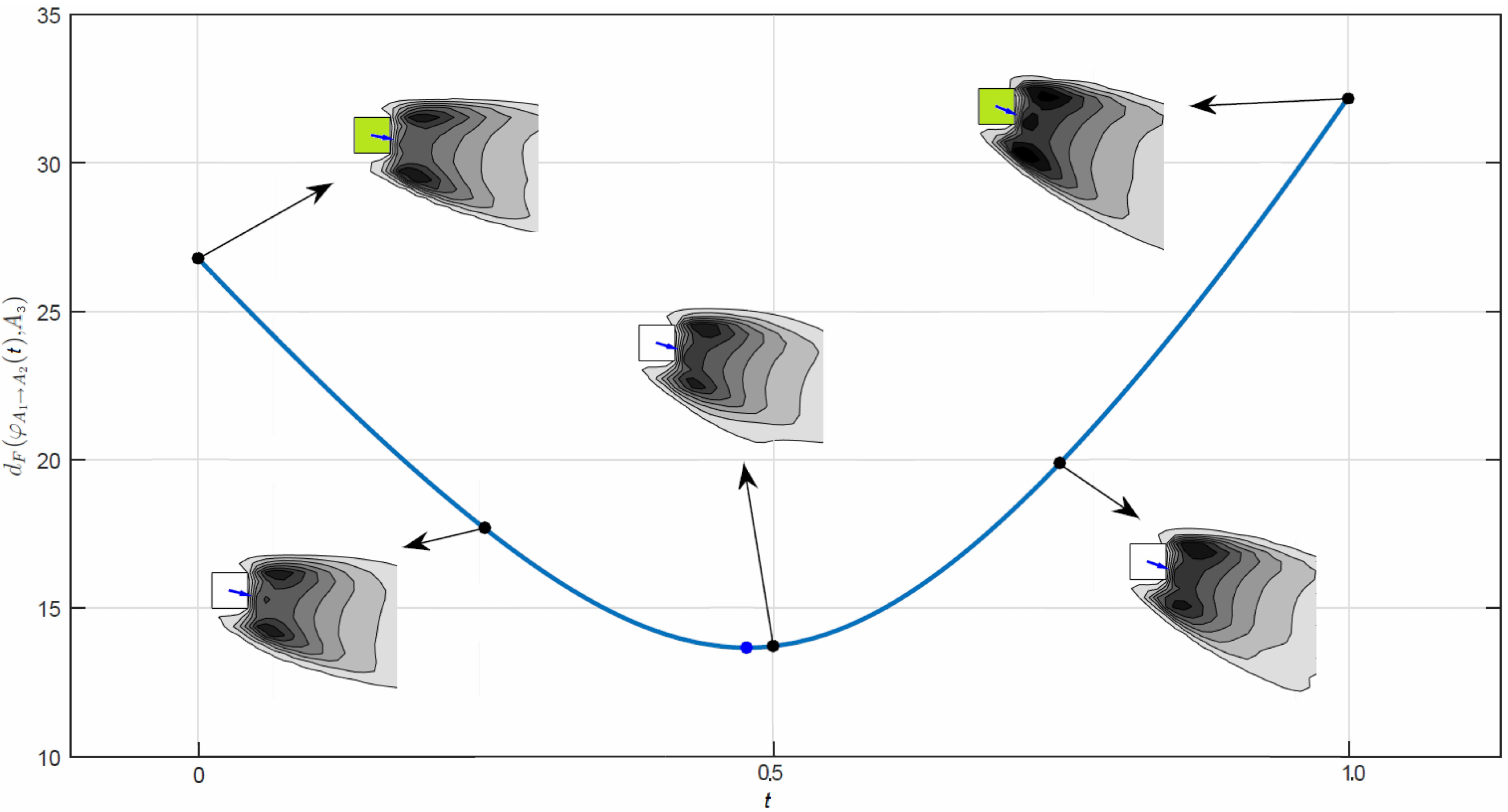}
  \caption{Distance from $A_3 = C(19.7,7)$ to the one-parameter covariance family built from $A_1 = C(16.9,7)$ and $A_2 = C(22.5,7)$, as a function of the input variable $t$. $A_{1}$ and $A_{2}$ are marked with green obstacles to identify them as the data matrices/anchors. The blue point represents the distance minimizer.}
  \label{fig:casestudy}
\end{figure}



%

\subsection{Distances to two-parameter covariance families}\label{sec:cs_twop}


Now we illustrate the distance between a given matrix and two different two-parameter covariance families, each constructed from the same four data matrices. The minimizer of this distance is a solution of  \Cref{prob:minsingular}.

First, we consider the first-order sectional covariance function of \Cref{sec:bilinear}. We use four data matrices: $A_1=C(11.3,4)$, $A_2=C(11.3,10)$, $A_3=C(16.9,4)$, $A_4=C(16.9,10)$. \Cref{fig:2pdistanceresult} (left) illustrates the distance from ${\varphi^{\mathrm{LS}}}_{(A_{1} \rightarrow A_{2})\rightarrow (A_{3}\rightarrow A_{4})}(t_1,t_2)$ to $A_5=C(14.1,7)$. The red triangle represents distance minimizer, which lies at $t_1=0.48$ and $t_2=0.77$ and yields a distance of $5.5$ from the target. To define the section in this case, we use the matrix $A_1$. The four anchor matrices are the edges of the rectangle in \Cref{fig:data}; other choices would lead to similar results, as analyzed in subsequent subsections.

Next, we repeat the study for the geodesic two-parameter covariance function defined in \Cref{sec:pgeodesic}, with results shown in \Cref{fig:2pdistanceresult} (right). Again, the minimizer is marked with a red triangle, which lies at $t_1=0.77$ and $t_2=0.48$ and yields a distance of $5.5$ from the target. Note that the inputs to both covariance functions can in principle be any element of $\mathbb{R}^2$; here, both figures show the distance for $(t_1, t_2) \in [0,2]\times [-1,1]$. The distance contours are shaped slightly differently for the sectional and geodesic cases, though the minimizer lies in the top left quadrant $[0,1]^2$ of each figure, as expected.


\begin{figure}[h!]
    \includegraphics[scale=.53]{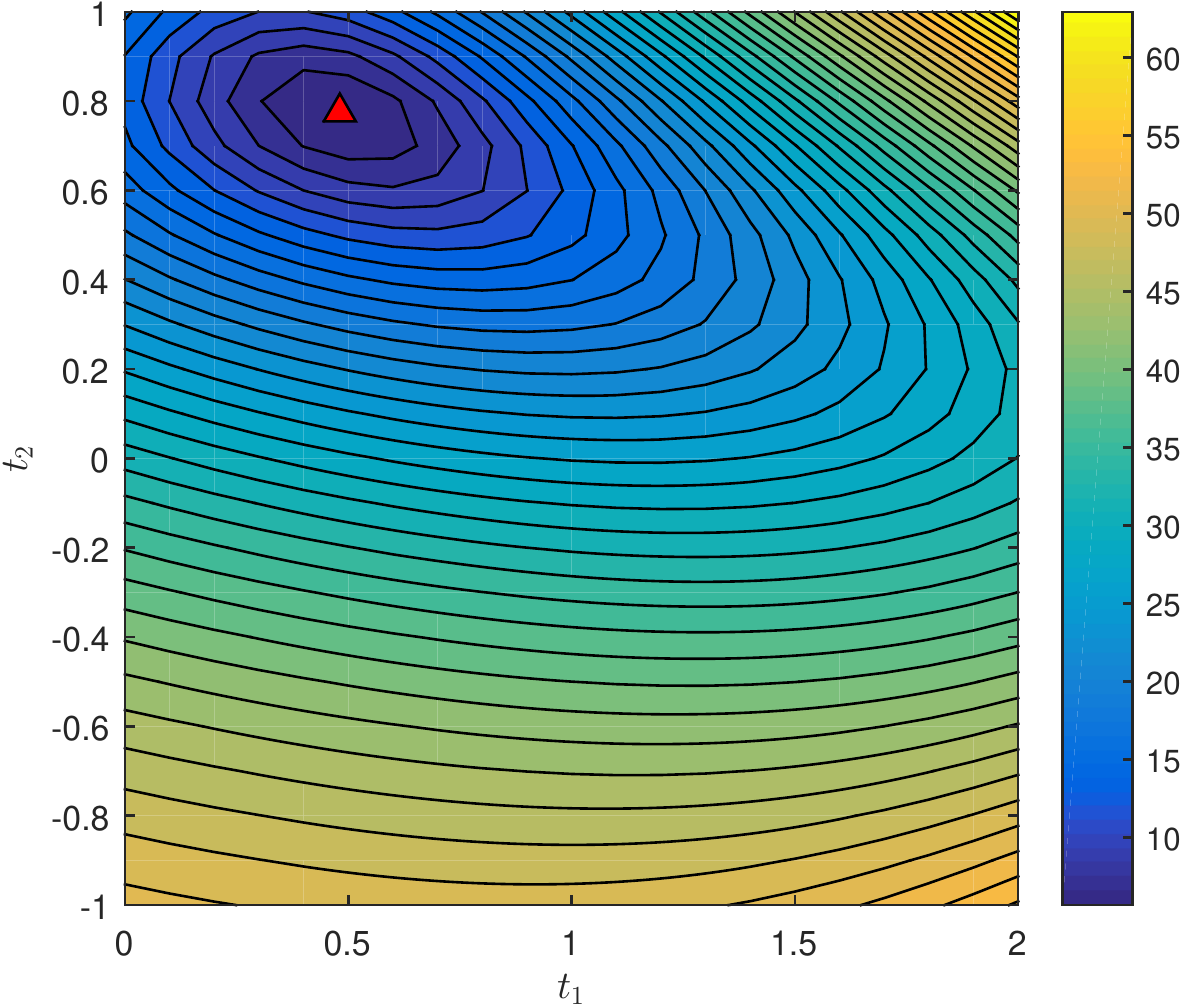}
    \includegraphics[scale=.53]{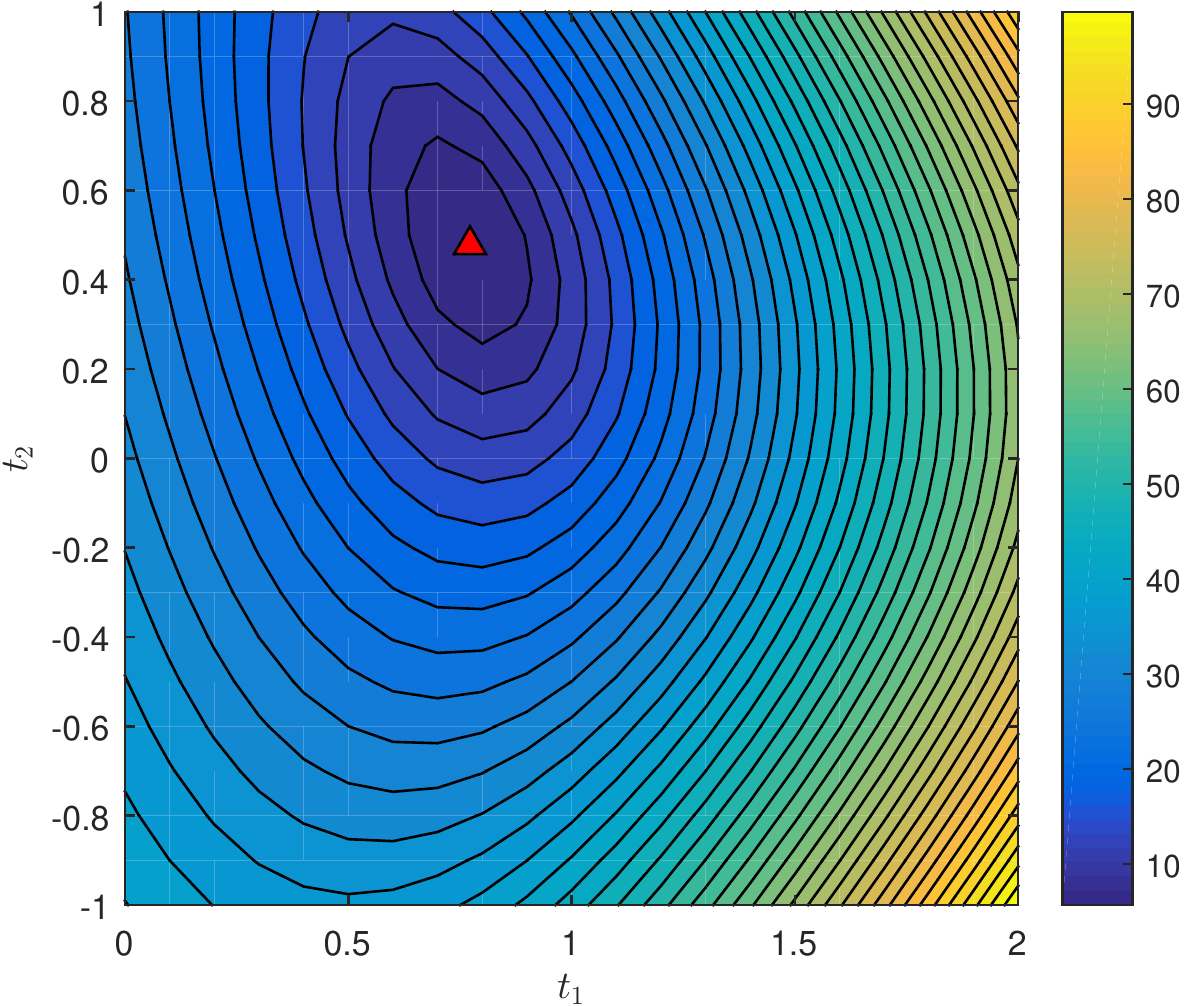}
\caption{Left: Distance from $A_5$ to ${\varphi^{\mathrm{LS}}}_{(A_{1} \rightarrow A_{2})\rightarrow (A_{3}\rightarrow A_{4})}(t_1,t_2)$. Right: Distance from $A_5$ to ${\varphi^{\mathrm{LG}}}_{(A_{1} \rightarrow A_{2})\rightarrow (A_{3}\rightarrow A_{4})}(t_1,t_2)$.}
\label{fig:2pdistanceresult}
\end{figure}

\subsection{Benchmarking first-order and higher-order covariance functions}


Now we consider the four surfaces defined in~\Cref{sec:surfaces}: the first-order covariance functions $\varphi^{\mathrm{LS}}$ and $\varphi^{\mathrm{LG}}$ defined patchwise (see Definition~\ref{def:1stOrderCovFunctions}), and the B\'ezier-like covariance functions $\varphi^{\mathrm{BS}}$ and  $\varphi^{\mathrm{BG}}$ (see Definitions~\ref{def:BezierCovFunctionSection} and~\ref{def:BezierCovFunctionManifold}).

For the surfaces defined on a section of the manifold, we consider several possibilities: for $\varphi^{\mathrm{LS}}$, the section based at one of the data matrices (here, the lower left data matrix of the patch), based at the arithmetic mean of the data matrices, or based at the inductive mean of the four data matrices of the patch. For $\varphi^{\mathrm{BS}}$, the section is based at one of the data matrices (here, the lower left data matrix of the training set), at the arithmetic mean of the data matrices of the training set, or at the inductive mean of the four data matrices of the training set.

These combinations lead to a total of eight surfaces. The data matrices are split into two sets shown in \Cref{fig:data}: the blue points and the black points.  The blue points are used to construct the surface, and the accuracy of the methods is evaluated on the black points.

 \begin{figure}[h!]
 \centering
 \includegraphics[scale=1]{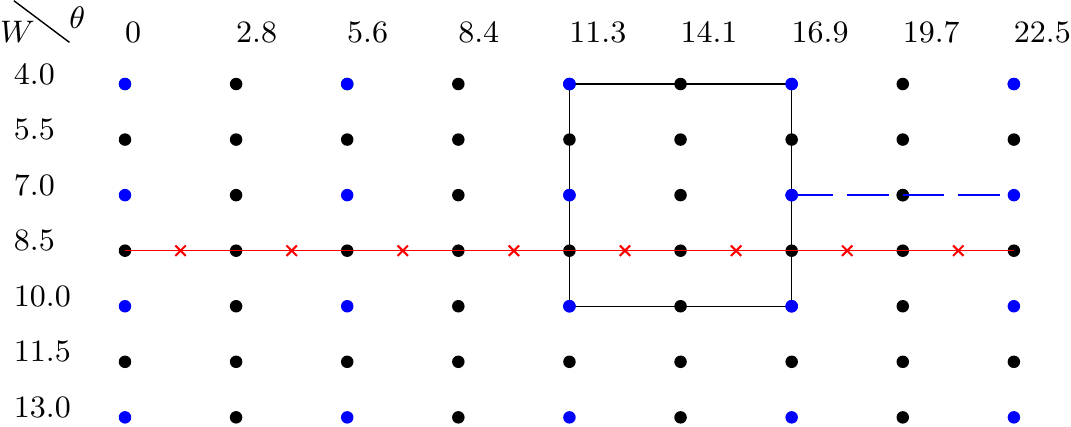}
 \caption{Data set for the wind field problem; each dot shows the wind magnitude $W$ and heading $\theta$ corresponding to a data covariance matrix. The red (crossed) and blue (dashed) lines represent the data used in \Cref{sec:cs_interpol}. The rectangle illustrates the operating zone of \Cref{sec:cs_twop}. The blue nodes comprise the training set and the black nodes the test set for \Cref{sec:interpol_2p,sec:compression}.}
 \label{fig:data}
 \end{figure}

\subsubsection{Interpolation errors}\label{sec:interpol_2p}
The error
\begin{equation*} E \left ( C(\theta, W) \right ) \coloneqq \left \|  C(\theta, W) - \varphi^{\text{method}}(\theta, W) \right \|^2_F \end{equation*}
is a measure of the ability of the surface $\varphi^{\text{method}}$ to recover some hidden covariance matrix $C(\theta, W)$.\footnote{Here we write $\varphi^{\text{method}}$ with arguments $(\theta, W)$ in a slight abuse of notation. More precisely, we mean that the two-parameter covariance function $\varphi^{\text{method}}$ is evaluated at $(t_1, t_2)$ corresponding to an affine mapping from the range of $(\theta, W)$ (here $[0,22.5] \times [4, 13]$) to $[0,4] \times [0,3]$, consistent with the $5 \times 4$ grid of data matrices.}
We will also consider the normalized error
\[ E_{{N}} \left ( C(\theta, W) \right ) \coloneqq 100 \times \frac{ \left \| C(\theta, W) - \varphi^{\text{method}}(\theta, W)\right  \|^2_F}{\frac{1}{4} \sum_{j = 1}^4 \left \|  C(\theta,W) - A_j \right \|^2_F}, \]
where normalization is performed with respect to the average squared distance between the target matrix $C(\theta, W)$ and the four corners of the patch to which it belongs, according to the grid representation in~\Cref{fig:data}.  (The patch is chosen systematically as the one below and to the left of the test point considered.)

Having defined these interpolation errors for arbitrary $C(\theta, W)$, we now evaluate them for all the points  $C(\theta_i, W_i)$ in the test set---i.e., for each of the data matrices marked with black nodes in \Cref{fig:data}. We average the errors over the test set and report the resulting values in~\Cref{table:interpresultscomparison}.


\begin{table}[!h]
\centering
\begin{tabular}{|p{4cm}|p{3cm}|p{3.3cm}|}
\hline
 & $\text{avg}_i \left [ E \left (C(\theta_i,W_i)\right ) \right ]$ & $\text{avg}_i \left [ E_{N} \left (C(\theta_i,W_i)\right ) \right ]$ \\
 \hline
1st-order section $\varphi^{\mathrm{LS}}_{\mathrm{one}}$  & 30.3 & 7.78\\
1st-order section $\varphi^{\mathrm{LS}}_{\mathrm{arithm}}$  & 30.4 & 7.80\\
1st-order section $\varphi^{\mathrm{LS}}_{\mathrm{inductive}}$  & 30.3  & 7.77\\
1st-order geodesic $\varphi^{\mathrm{LG}}$ & 30.3  & 7.77\\
 B\'ezier  section $\varphi^{\mathrm{BS}}_{\mathrm{one}}$ & 20.8  & 4.91 \\
 B\'ezier  section $\varphi^{\mathrm{BS}}_{\mathrm{arithm}}$ & 20.5 & 4.87 \\
 B\'ezier  section $\varphi^{\mathrm{BS}}_{\mathrm{inductive}}$ & 20.4  & 4.85\\
 B\'ezier geodesic $\varphi^{\mathrm{BG}}$& 20.3  & 4.79 \\
\hline
\end{tabular}
\caption{Average (squared) distance separating a given test point $C(\theta_i, W_i)$ from the corresponding interpolation point on the different surfaces. For the methods defined on a section of the manifold, the subscript of $\varphi$ indicates how the section was chosen: `one' means that we use one of the data matrices (here, the matrix at the lower left corner of the patch) as basis of the section, while `arithm' and `inductive' denote, respectively, the arithmetic and inductive means of the corners of the patch.}
\label{table:interpresultscomparison}
\end{table}

Some takeaways from this study are as follows. First, the matrix chosen to define the section in either the first-order sectional covariance function or the higher-order sectional covariance function seems to have little impact. Moreover, the performance of the geodesic covariance functions is not noticeably better than that of the sectional functions in this setting. But the interpolation performance of the higher-order (B\'{e}zier) families is significantly better than that of the first-order families.



\subsubsection{Identification errors and data compression}\label{sec:compression}

We now assess identification errors within the covariance families. In other words, we now use the techniques of \Cref{sec:results} to minimize the distance from each element of the test set to the covariance family $\varphi^{\text{method}}$, built patchwise from the training matrices. From another perspective, this process can be viewed as data compression: a simple way to perform data compression consists of storing only several matrices (in our case, the training data) and then storing, for any additional matrix, the coordinates of the closest point in the surface. We now compare our surfaces for this task. Similarly to the previous section, we use the following two error measures,
\begin{align*}
     E^* \left ( C(\theta_i,W_i) \right ) &=  \left  \| C(\theta_i,W_i) - \varphi(t_1^*(i),t_2^*(i))\right  \|^2_F, \\
     E^*_{{N}} \left ( C(\theta_i,W_i) \right ) &= \frac{ 100  \times \left  \| C(\theta_i,W_i) - \varphi(t_1^*(i),t_2^*(i)) \right  \|^2_F}{\frac{1}{4} \sum_{j = 1}^4 \left \|  C(\theta_i,W_i) - A_j \right \|^2_F},
\end{align*}
where $A_j$ are the four corners of the patch to which $C(\theta_i,W_i)$ belongs and $t_1^*(i),t_2^*(i)$ are the solutions to the optimization problem discussed in Section~\ref{sec:results}.

We evaluate these errors for every test matrix and report, in \Cref{table:august1resultscomparison}, the average errors for each surface definition proposed. These are essentially the average distance between an element of our test set and the closest point on the surface. A key takeaway from this table is that the error are significantly lower than those in \Cref{table:interpresultscomparison}; this is not surprising, as here we are optimizing to find the best point in each family. Also, results with the geodesic families in this example appear to be slightly better than with the sectional families.

\begin{table}[!h]
\centering
\begin{tabular}{|p{4cm}|p{3.3cm}|p{3.3cm}|}
\hline
& $\text{avg}_i \left [ E^* (C(\theta_i,W_i) ) \right ]$ &
  $\text{avg}_i \left [ E^*_{N} \left (C(\theta_i,W_i)\right ) \right ]$ \\
 \hline
1st-order section $\varphi^{\mathrm{LS}}_{\mathrm{one}}$ & 21.8 & 5.30\\
1st-order section $\varphi^{\mathrm{LS}}_{\mathrm{arithm}}$  & 21.8 & 5.29\\
1st-order section $\varphi^{\mathrm{LS}}_{\mathrm{inductive}}$  & 21.8  & 5.30\\
1st-order geodesic $\varphi^{\mathrm{LG}}$& 20.9  & 5.10\\
 B\'ezier  section $\varphi^{\mathrm{BS}}_{\mathrm{one}}$ & 14.3  & 3.34 \\
 B\'ezier  section $\varphi^{\mathrm{BS}}_{\mathrm{arithm}}$ & 14.0  & 3.29 \\
 B\'ezier  section $\varphi^{\mathrm{BS}}_{\mathrm{inductive}}$ & 14.0  & 3.29\\
 B\'ezier geodesic $\varphi^{\mathrm{BG}}$& 13.8  & 3.24\\
\hline
\end{tabular}
\caption{Average (squared) distance separating a given test point from its closest approximation on the different surfaces. For methods defined on a section of the manifold, the subscript of $\varphi$ indicates how the section was chosen: `one' means that we use one of the data matrices (here, the matrix at the lower left corner of the patch) as basis of the section, while `arithm' and `inductive' denote, respectively, the arithmetic and inductive means of the corners of the patch.}
\label{table:august1resultscomparison}
\end{table}

It is instructive to see how the normalized errors $E^*_{{N}} \left ( C(\theta_i,W_i)\right )$ are distributed over the data set, i.e., how the approximation error depends on the parameters of the data matrices $C(\theta_i,W_i)$. We illustrate this distribution using the stem plot in \Cref{fig:distrError}, for $\varphi^{\text{BG}}$ only. Errors are largest for wind field headings between those of the training set (i.e., $\theta \in \{2.8, 8.4, 14.1, 19.7\}$) and increase strongly with the wind field magnitude. These trends indicate that it might be useful to generate a denser grid of data matrices in the $\theta$ direction, particularly for large $W$, and to use a coarser grid in the $W$ direction.

\begin{figure}[h!]
\centering
\includegraphics[scale=0.5]{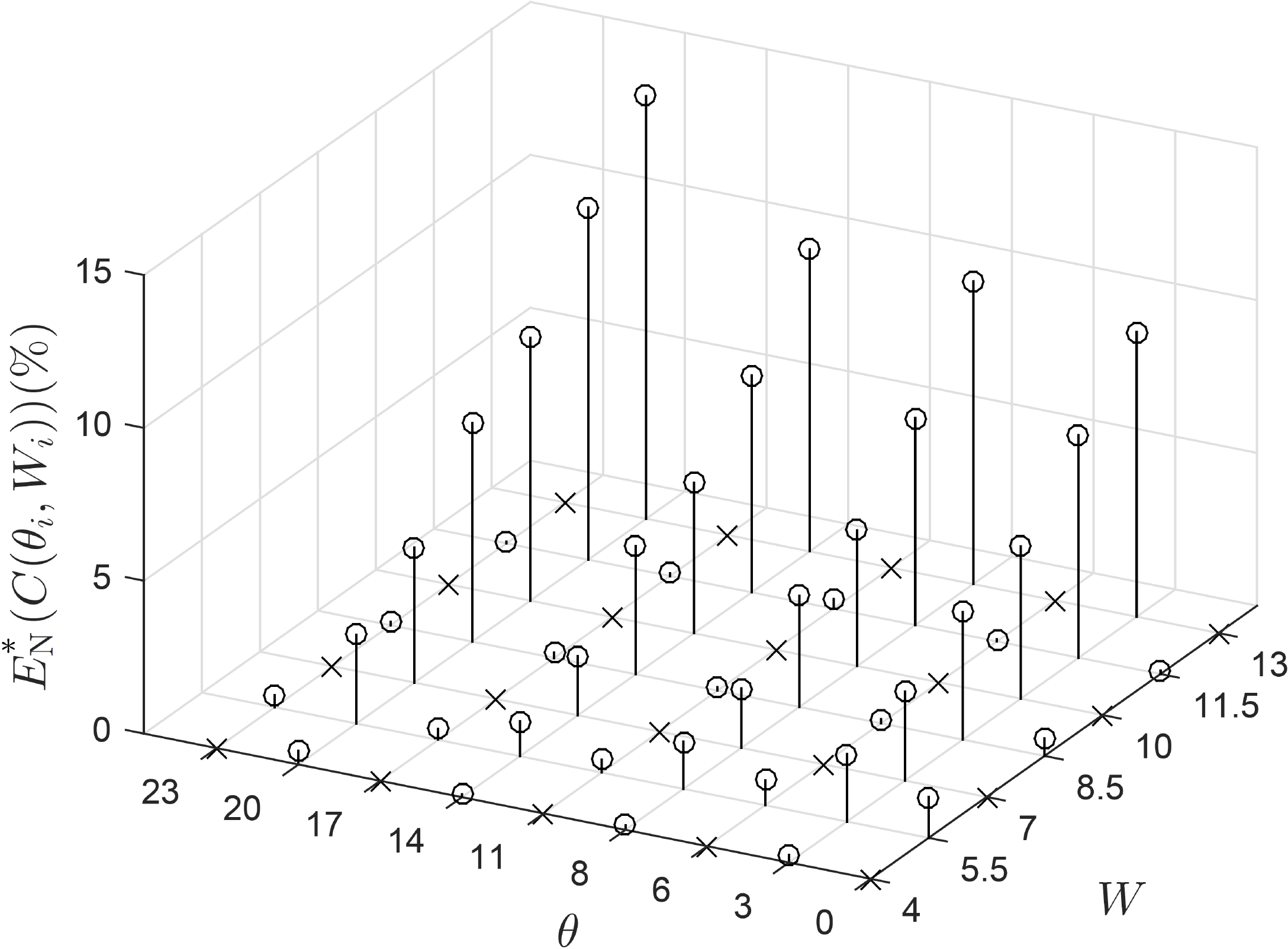}
\caption{Distribution of the errors obtained with the higher-order geodesic covariance family $\varphi^{\mathrm{BG}}$. Crosses are training points, circles are test points. It is more difficult to recover the data when varying $\theta$, particularly at larger $W$. Interpolation in the $W$ direction, on the other hand, yields very small errors.}
\label{fig:distrError}
\end{figure}

%
%

\section{Conclusions}\label{sec:conclu}
We have presented a differential geometric framework for constructing parametric low-rank covariance families, by connecting low-rank covariance matrices obtained at representative problem instances. In this sense, our framework creates {parametric} covariance families that can easily incorporate prior knowledge or {empirical} information, via the given data matrices or anchors. Within this broad framework, we have proposed several different constructions that rely on geodesics on the manifold of positive-semidefinite matrices, or on affine sections of this manifold. We presented particular instances of such covariance functions that interpolate grids of data matrices, using either first- or higher-order (B\'{e}zier) approaches.

Given some data and the resulting sample covariance matrix (which can be strongly rank-deficient) we show how to perform minimum-distance covariance identification, that is, how to find the closest element of a given covariance family. We discuss methods and algorithms to solve this problem for each of the proposed covariance functions. In a case study involving wind velocity field approximation, we assess the ability of our covariance families to represent out-of-family covariance matrices. We also demonstrate the possibility of using this technique for data compression, i.e., storing the parameters of the family corresponding to a particular matrix, instead of the matrix itself.

Moreover, if the data matrices are labelled by some characteristics of the problem, we observe that our covariance functions achieve a natural or desirable behavior---in that the inputs to these functions (related to distance along a geodesic or a sectional projection) are closely related to the label parameters themselves. For instance, our case study shows that a family constructed from covariance matrices obtained at very different wind headings contains matrices that match all intermediate headings, and that the input parameter to the covariance function and the wind heading can be nearly linearly related.

An advantage of the proposed framework is that one can choose the rank of the resulting matrices. To illustrate, in the case study, although the dimension of the covariance matrices is high ($n=3024$), we can reduce the rank to $r=20$, saving considerable computational cost while still obtaining good approximations.

We also note that this paper has focused on the definition and construction of low-rank covariance families, and on efficient optimization methods for solving identification problems within these families. An interesting and complementary line of work could analyze certain associated statistical questions, i.e., statistical properties of the minimum distance estimator as a function of the sample size $q$, the matrix dimension $n$, and the chosen rank $r$. We defer such investigations to future work.

\section*{Implementation code}
The code to generate covariance functions and to perform identification and interpolation is available at \url{https://github.com/EMassart/covariance_fitting}.


\appendix

\section{Computation of the tools required for the variable projection method}\label{ap:variableprojection}
We detail here the computations for the main steps of the variable projection method proposed in \Cref{sec:variablemin}. Remember that the corresponding surface, $\varphi^{\mathrm{LG}}(t_1,t_2) = Y_{\varphi}(t_1,t_2) Y_{\varphi}(t_1,t_2)^\T$, where $Y_{\varphi}$ is obtained by composition of two geodesics, respectively along the $t_1$ and $t_2$ variables:
\begin{align*}
    Y_{\varphi}(t_1, t_2) &\coloneqq (1-t_2) Y_{1-2}(t_1) + t_2 Y_{3-4}(t_1)Q(t_1)^\T \:, \\
    Y_{1-2}(t_1) &\coloneqq (1-t_1) Y_1 + t_1 Y_2 Q_{1-2}^\T \:, \\
    Y_{3-4}(t_1) &\coloneqq (1-t_1) Y_3 + t_1 Y_4 Q_{3-4}^\T \:.
\end{align*}

\subsection{Computation of the partial derivative with respect to $t_1$}
Computing the partial derivative of the function $(t_1,t_2) \mapsto \varphi^{\mathrm{LG}} (t_1,t_2)$ with respect to $t_1$ yields:
\[D \varphi^{\mathrm{LG}} (t_1,t_2) [e_1] = D Y_{\varphi} (t_1,t_2)[e_1] \ Y_{\varphi} (t_1,t_2)^\T +  Y_{\varphi} (t_1,t_2) \ D Y_{\varphi} (t_1,t_2)[e_1]^\T, \]
with
\[D Y_{\varphi} (t_1,t_2)[e_1] = (1-t_2) \dot Y_{1-2}(t_1)  + t_2 \dot  Y_{3-4}(t_1) Q(t_1)^\T + t_2 Y_{3-4}(t_1) \dot  Q(t_1)^\T. \]
The values of $\dot Y_{1-2}(t_1)$ and $\dot Y_{3-4}(t_1)$ are independent of $t_1$:
\[ \dot Y_{1-2}(t_1) = - Y_1 + Y_2  Q_{1-2}^\T\: , \qquad \dot Y_{3-4}(t_1) = - Y_3 + Y_4  Q_{3-4}^\T\:, \qquad \forall t_1. \]
The value of $\dot Q(t_1)$ can be obtained as follows. Recall from the geodesic definition that $Q(t_1)$ is the orthogonal factor of the polar decomposition of the matrix $M(t_1) = Y_{1-2}(t_1)^\T Y_{3-4}(t_1),$ which means that there exists a symmetric positive definite matrix $H(t_1)$ such that $M(t_1) = H(t_1)Q(t_1)$. Then,
\begin{equation}
\dot M(t_1) = \dot H(t_1) Q(t_1) + H(t_1) \dot Q(t_1),
\label{eq:Mt1}
\end{equation}
where $\dot H(t_1)$ is a symmetric matrix, and $\dot Q(t_1)$ is of the form $\dot Q(t_1) =  \Omega(t_1) Q(t_1) $, with $\Omega(t_1)$ a skew-symmetric matrix. Right-multiplying this expression by $Q(t_1)^\T$ yields:
\begin{equation}
\dot M(t_1) Q(t_1)^\T = \dot H(t_1)  + H(t_1) \Omega(t_1),
\label{eq:Mt1QT}
\end{equation}
while left-multiplying the transpose of equation~\eqref{eq:Mt1} by $-Q(t_1)$ yields:
\begin{equation}
- Q(t_1) \dot M(t_1)^\T = - \dot H(t_1) + \Omega(t_1) H(t_1).
\label{eq:QMt1}
\end{equation}
Now, summing equations~\eqref{eq:Mt1QT} and~\eqref{eq:QMt1} yields:
\[ \dot M(t_1) Q(t_1)^\T - Q(t_1) \dot M(t_1)^\T =  H(t_1) \Omega(t_1) +\Omega(t_1) H(t_1). \]
As a result, the term $\dot Q_1(t_1)$ can be obtained by solving a Sylvester equation.  Moreover, since $H(t_1)$ is always positive definite (except in the set of zero measure corresponding to low-rank matrices $Y_{1-2}(t_1)^\T Y_{3-4}(t_1)$), the solution to the Sylvester equation is unique ($H(t_1)$ and -$H(t_1)$ have no common eigenvalues).
\subsection{Computation of $t_2^*(t_1)$}
The first-order optimality condition
\[ \left.\frac{\partial f}{\partial t_2}\right|_{(t_1 t_2^*(t_1))} = 0 \]
implies that the optimal value $t_2^*(t_1)$ corresponding to an arbitrary value $t_1$ is the solution to a cubic equation:
\begin{equation}
s_1(t_1) t_2^3(t_1) + s_2(t_1) t_2^2(t_1) + s_3(t_1) t_2(t_1) + s_4(t_1) = 0,
\label{eq:solY2Opt}
\end{equation}
with
\begin{align*}
& s_1(t_1) = 2 \Tr{R^2} = 2 \sum_i \sum_j R_{ij}^2, \\
& s_2(t_1) = 3 \Tr{R S} = 3 \sum_i \sum_j R_{ij}S_{ij}, \\
& s_3(t_1) = 2 \Tr{R T} + \Tr{S} = 2 \sum_i \sum_j R_{ij}T_{ij} + S_{ij}^2 ,\\
& s_4(t_1) = 2 \Tr{S T} =  2 \sum_i \sum_j S_{ij}T_{ij}.
\end{align*}
The matrices $R$, $S$, and $T$ arising in those expressions are defined as:
\begin{align*}
R &= Y_{1-2}Y_{1-2}^\T + Y_{3-4}Y_{3-4}^\T - \left( Y_{1-2} Q Y_{3-4}^\T + Y_{3-4} Q^\T Y_{1-2}^\T \right), \\
S &=  \left( Y_{1-2} Q Y_{3-4}^\T + Y_{3-4} Q^\T Y_{1-2}^\T \right) -  2  Y_{1-2}Y_{1-2}^\T\;, \\
T &=  Y_{1-2}Y_{1-2}^\T - \widehat{C},
\end{align*}
with all these matrices depending on $t_1$. Observe, however, that for a fixed value of $t_1$, the function $ t_2 \rightarrow f(t_1,t_2)$ might not be convex; hence, the condition~\eqref{eq:solY2Opt} might have several (up to three) real solutions. In that case, we compute the value of the cost function at those solutions, and we choose the one corresponding to the smallest value of $f$.

\subsection{Gradient descent for the univariate cost function}
We are now looking for the derivative of the cost function $\tilde f(t_1) = f(t_1, t_2^*(t_1))$, with respect to the variable $t_1$, in order to be able to apply a steepest descent method to that problem. Using the notation $\tilde f = F \circ \tilde {\varphi}^{\mathrm{LG}}$, with $\tilde{ \varphi}^{\mathrm{LG}}(t_1) =  \varphi^{\mathrm{LG}}(t_1,t_2^*(t_1))$, we have:
\[ \dot{\tilde f}(t_1) = DF[ \dot{\tilde{\varphi}}^{\mathrm{LG}} (t_1)] = 2 \Tr{\dot{\tilde{\varphi}}^{\mathrm{LG}}(t_1) (\tilde {\varphi}^{\mathrm{LG}}(t_1) - \widehat{C})^\T}. \]
The derivative $\dot{\tilde{\varphi}}^{\mathrm{LG}}(t_1)$ is given by:
\begin{equation*}
\dot{\tilde{\varphi}}^{\mathrm{LG}}(t_1) = \dot {Y}_{\tilde{ \varphi}}(t_1) Y_{\tilde{\varphi}}(t_1)^\T + Y_{\tilde{\varphi}} \dot{Y}_{\tilde {\varphi}}(t_1)^\T.
\end{equation*}
Using the chain rule,
\begin{equation*}
\dot{Y}_{\tilde {\varphi}}(t_1) = \frac{\partial Y_{\varphi}}{\partial t_1}(t_1) +   \frac{\partial Y_{ \varphi}}{\partial t_2}(t_2^*(t_1)) \dot t_2^*(t_1).
\end{equation*}
By the definition of $t_2^*(t_1)$, the term $\frac{\partial Y_{\gamma}}{\partial t_2}(t_2^*(t_1))$ is equal to zero. As a result, $\dot{Y}_{\tilde \varphi}(t_1) = \frac{\partial Y_{\varphi}}{\partial t_1}(t_1)$, which has been computed earlier.

\bibliographystyle{siamplain}
\bibliography{references}
\end{document}